\documentclass[12pt,draftclsnofoot,journal,onecolumn]{IEEEtran}

\usepackage{amsmath}
\usepackage{amssymb}
\usepackage{graphicx}
\usepackage{bm}
\usepackage{paralist}
\usepackage{subfigure}
\usepackage{algorithm}
\usepackage{color}

\newcommand{\bx}{\bm{x}}

\newcommand{\br}{\ensuremath{{\bm r}}}
\newcommand{\by}{\ensuremath{{\bm y}}}

\newcommand{\bs}{\ensuremath{{\bm s}}}

\newcommand{\set}[1]{\ensuremath{\mathcal #1}}

\newcommand{\separator}{
  \begin{center}
    \rule{\columnwidth}{0.3mm}
  \end{center}
}

\def\eg{{\it e.g.}}
\def\ie{{\it i.e.}}

\newtheorem{theorem}{Theorem}[section]

\newtheorem{lemma}{Lemma}[section]
\newtheorem{definition}{Definition}[section]

\newcommand{\beq}{\begin{eqnarray*}}
\newcommand{\eeq}{\end{eqnarray*}}
\newcommand{\beqn}{\begin{eqnarray}}
\newcommand{\eeqn}{\end{eqnarray}}
\newcommand{\bemn}{\begin{multiline}}
\newcommand{\eemn}{\end{multiline}}

\newcommand{\sqeq}{\addtolength{\thinmuskip}{-4mu}
\addtolength{\medmuskip}{-4mu}\addtolength{\thickmuskip}{-4mu}}

\newcommand{\unsqeq}{\addtolength{\thinmuskip}{+4mu}
\addtolength{\medmuskip}{+4mu}\addtolength{\thickmuskip}{+4mu}}

\def\N{\mathbb{N}}
\def\R{\mathbb{R}}

\newcommand{\fbs}[1]{\bfseries\slshape #1}

\newcommand{\slink}{V}

\newcommand{\blambda}{\boldsymbol{\lambda}}

\newcommand{\bsigma}{\boldsymbol{\sigma}}

\renewcommand{\N}{\mathcal{N}}

\newcommand{\bor}{\boldsymbol{r}}

\newcommand{\sI}{\mathcal{I}}

\newcommand{\y}{\boldsymbol{y}}
\newcommand{\E}{\mathcal{E}}

\newcommand{\n}{\mathcal{N}}

\begin{document}
%
 \title{CSMA using the Bethe Approximation: Scheduling and Utility Maximization\thanks{A part of this work was presented at IEEE ISIT 2013.}}

 \author{Se-Young Yun, Jinwoo Shin, and Yung Yi 
   
   \thanks{S. Yun is with
     MSR-INRIA, France (e-mail:
     seyoung.yun@inria.fr). J. Shin and Y. Yi are with the Department of Electrical
     Engineering, KAIST, Korea (e-mail: jinwoos@kaist.ac.kr ,yiyung@kaist.edu).}  }

\maketitle

\begin{abstract}
  CSMA (Carrier Sense Multiple Access), which resolves contentions over
  wireless networks in a fully distributed fashion, has recently gained
  a lot of attentions since it has been proved that appropriate control
  of CSMA parameters guarantees optimality in terms of stability (\ie,
  scheduling) and system-wide utility (\ie, scheduling and congestion
  control). Most CSMA-based algorithms rely on the popular MCMC (Markov Chain Monte
  Carlo) technique, which enables one to find optimal CSMA parameters
  through iterative loops of `simulation-and-update'.  However,
  such a simulation-based approach often becomes a major cause of
  exponentially slow convergence, being poorly adaptive to flow/topology
  changes.  In this paper, we develop distributed iterative algorithms
  which produce approximate solutions with convergence in polynomial
  time for both stability and utility maximization problems. In
  particular, for the stability problem, the proposed distributed
  algorithm requires, somewhat surprisingly, only one iteration among
  links.  Our approach is motivated
  by 
  the Bethe approximation (introduced by Yedidia, Freeman and Weiss \cite{Yedidia05constructingfree}) allowing us to express approximate solutions via a certain
  non-linear system with polynomial size. Our polynomial convergence
  guarantee comes from directly solving the non-linear system in a
  distributed manner, rather than multiple simulation-and-update loops
  in existing algorithms.  We provide numerical results to show that the
  algorithm produces highly accurate solutions and converges much faster
  than the prior ones.
\end{abstract}

\section{Introduction}

\subsection{Motivation}

Recently, it has been proved that CSMA, albeit simple and fully
distributed, can achieve high performance in terms of throughput (\ie, the stability problem) and fairness (\ie, the utility maximization problem)
by joint scheduling/congestion controls \cite{JSSW10DRA,
  JW10DC,IEEEpaper:Liu_Yi_Proutiere_Chiang_Poor_2009,HP12SO}. These
advances show that even an algorithm with no or little message passing
can actually be close to the optimal performance, achieving significant progress
in terms of algorithmic complexity over the seminal work of Max-Weight
\cite{taseph92} and its descendant researches, each of which often
takes a tradeoff point between complexity and performance, see
\cite{LSS06,yi2008stochastic}. The main idea underlying the recent CSMA
developments is to intelligently control access intensities (i.e.,
access probability and channel holding time) over links so as to let
the resulting long-term link service rate converge to the target rate \cite{YY12OC}.

However, one of the main drawbacks for such CSMA algorithms is slow
convergence, which is problematic in practice due to its poor
adaptivity to network and flow configuration changes. The root cause
of slow convergence stems from the fact that all the above algorithms
are based on the MCMC (Markov Chain Monte Carlo) technique, where even
for a fixed CSMA intensity, it takes a long time, called mixing time, to
reach the stationary distribution to observe how the system
behaves. Note that the mixing time is typically exponentially large
with respect to the number of links \cite{JLN11FM}. 
For the mixing time issue, there
exist algorithms updating CSMA intensities before the system is mixed,
e.g., without time-scale separation between the intensity update and
the time to get the system state for a given intensity update
\cite{IEEEpaper:Liu_Yi_Proutiere_Chiang_Poor_2009,HP12SO}.  However,
they are not free from the slow convergence issue since their
convergence inherently also requires the mixing property of the
underlying network Markov process.  In summary, all prior CSMA
algorithms suffer from slow convergence explicitly or implicitly. The
main goal of this paper is to develop `mixing-independent' CSMA
algorithms to overcome the issue at the marginal cost of performance
degradation.

\subsection{Goal and Background}
We aim at drastically improving the convergence speed by using
the techniques in artificial intelligence and statistical physics (instead of the MCMC based ones)
for both stability under unsaturated case and utility
maximization under saturated case. For instance, in order to reach the convergent
service rates as the solution of the utility maximization problem,
the intermediate target service rates should be iteratively updated
toward the optimal rates, from which the transmission intensities are
consequently updated. Our key contribution lies in designing message-passing mechanisms to directly compute the required
access intensity for given target service rates in a distributed
manner, rather than estimation-based approaches in the MCMC
technique. 
In what follows, we present some necessary backgrounds before 
we describe more details of our main contributions

The CSMA setting can be naturally understood by a certain Markov random
field (MRF) \cite{kindermann1980markov}, which we call CSMA-MRF, in the
domain of physics and probability. In CSMA-MRF, links induce a graph
where links are represented by vertices and interfering links generate
edges. Access intensities over links correspond to MRF-parameters in
CSMA-MRF. Then, the service rate of each link becomes the marginal
distribution of the corresponding vertex in CSMA-MRF. In the area of
MRFs, free energy concepts such as `Gibbs free energy' function and
`Bethe free energy' function defined by the graph and MRF-parameter have
been studied to compute marginal probabilities in MRFs. For example, it
is known by \cite{Yedidia05constructingfree} that finding a minimum (or
zero-gradient) point of a Bethe function can lead to approximated values
for marginal distributions, where its empirical success has been widely
evidenced in many areas such as computer vision, artificial intelligence
and information
theory~\cite{Yedidia05constructingfree,F01CG,MWJ99LB}. The main benefit
of this approach is that zero-gradient (non-linear) equations of a free
energy function can provide low-complexity (approximate) consistency
conditions between marginal probabilities and MRF-parameters.

\subsection{Contribution}


First, for the stability problem, we assume that each link is
  aware of only its {\em local} load, \ie, its targeted marginal
  probability in CSMA-MRF.\footnote{The knowledge about the local
    (offered) load may be learnt by empirical estimations or provided
    by the admission control of the incoming flows.}  Given targeted
  marginal probabilities, we show that the Bethe equation
  (corresponding to the stability problem) is solvable, somewhat
  surprisingly, in one iteration among links.  Equivalently, each link
  can calculate its approximate access intensity for targeted
  throughputs of links in one iteration of message-passings between
  neighbors. The result relies on the following special property of
  CSMA-MRF (which is not applicable for other general MRFs): 
\begin{itemize}
  \item[($\dagger$)]The higher-order marginal probabilities needed by
    the Bethe free energy (BFE) 
  functions are decided by the first-order marginal probabilities in
  CSMA-MRF.
\end{itemize}
Our algorithm, called {\bf BAS}, for the
  stability problem are presented in Section \ref{sec:stability}.
  
Second, we provide a distributed CSMA algorithm, called {\bf BUM},
  for the utility maximization problem, and show that it converges in a polynomial
  number of iterations, which is dramatically faster than prior algorithms based on MCMC.
  The {\bf BUM} algorithm consists of two phases: the first and second
phases aim at computing targeted service rates (\ie, marginal distributions) and corresponding 
CSMA intensities (\ie, MRF parameters),
respectively. 
We formulate these computational
problems as minimizing  Bethe free energy (BFE)
functions. We show that the Bethe function in its first phase is
convex for the popular $\alpha$-fairness utility functions \cite{mowa00}, and
develop a distributed gradient algorithm for minimizing it.
For the second phase, we use the {\bf BAS} algorithm developed for the
stability problem. We also
characterize the error of the {\bf BUM} algorithm in terms of that of
the {\bf BAS} algorithm, i.e., if {\bf BAS} is accurate, {\bf BUM} is as
well. The description and analysis of {\bf BUM} are given in Section
\ref{sec:utilmax}.

Our main technical contribution for the {\bf BUM} algorithm lies in developing
 a distributed gradient algorithm in the first phase. Even though we prove that 
the BEF function is convex,
it is still far from being clear
  that a distributed gradient algorithm can achieve its minimum since
  its domain is a bounded polytope, i.e., the BFE function is
  constrained by linear inequalities. To overcome this issue, we use the
  following special property of the BFE function in CSMA-MRF (which is
  not generally true for
  other BFE functions): 
\begin{itemize}
	\item[($\ddagger$)] The minimum of the BFE function is strictly inside of its domain. 
\end{itemize}
Using the property ($\ddagger$), we carefully choose a (dynamic)
projection scheme for the gradient algorithm so that it never hits the
boundary of the BFE function after a number of iterations, say $T$.
Then, after $T$ iterations, the gradient algorithm is analyzable to
converge similarly as its optimizing function is unconstrained.  

\smallskip Our simulation results show that the proposed schemes
converge fast and the approximation is accurate enough. First, we test
the actual service rate of {\bf BAS}  and verify that the
service rates are close to the target rates. Next, {\bf BUM} is
compared with conventional utility optimal CSMA algorithms. In the
results,  {\bf BUM} converges within 1000 iterations, whereas the
conventional schemes do not converge even until 10000
iterations. Moreover, the achieved network utility is almost the same with the
utility by conventional algorithms. We also note that {\bf
    BUM} can converge much quickly. Since each update of {\bf BUM} does not require
  to estimate the underlying service rates, we can run {\bf
    BUM} as an offline algorithm which can be done without any packet transmission.

\smallskip
In addition to MCMC-based approaches on developing CSMA algorithm for the stability and utility maximization problems,
the authors of \cite{KL12ABP} studied 
the Belief Propagation (BP)
algorithm for solving them. BP and BFE functions are connected as discussed
in \cite{Yedidia05constructingfree}, in that there is an one-to-one correspondence between
fixed points of BP and local minima of BFE functions.
However, the proposed algorithms in \cite{KL12ABP} may take a long time to converge
for the stability problem, and may
not converge at all for the utility maximization problem. 
Our work differs from \cite{KL12ABP} in that BFE
functions are exploited not to find marginal distributions in
CSMA-MRF but to find MRF-parameters given the targeted
marginal distributions.

\section{Model and Problem Description}
\label{sec:model}
For reader's convenience, we make a list of notations, which is given in Appendix~\ref{apd:nota}
\subsection{Model}


\noindent{\fbs Network model.}
In a wireless network, each link $i$, which consists of a transmitter node
and a receiver node, shares the wireless medium with its `neighboring'
links, meaning the ones that are interfering
with $i$, i.e., the transmission over $i$ cannot be successful if a
transmission in at least one neighboring link occurs simultaneously.  
We assume that each link has a unit capacity. 
The interference relationship among links can be represented by
a graph $G=(\slink,E)$, popularly known as the {\em interference
  graph}, where links in the
wireless network are represented by the set of vertices $\slink,$ and
any two links $i,j$ share an edge $(i,j)\in E$ if their transmissions interfere with each other. 


\smallskip
\noindent{\fbs Feasible rate region.}
We let $\bsigma(t) \triangleq [\sigma_i(t) \in \{0,1\}:i\in\slink]$~\footnote{Let $[x_i:i\in\slink]$ denote the vector whose
  $i$-th element is $x_i.$ For notational convenience, instead of
  $[x_i:i\in\slink]$, we use $[x_i]$ in the remainder of this paper.} 
denote the {\em scheduling vector} at time $t$, i.e., link $i$ is
active or transmits packets (if it has any)
with unit rate at time $t$ if $\sigma_i(t)=1$ (and does not otherwise). 
The scheduling vector $\bsigma(t)$ is said to be {\em feasible} if no interfering links are active simultaneously at time $t$, 
i.e., $\sigma_i(t)+ \sigma_j(t) \le 1,~\forall (i,j) \in E$. 
We use $\N(i)\triangleq \{j: (i,j)\in E \}$ to
denote the set of the neighboring links of link $i$, $d(i) \triangleq |\N(i)|$ and $d \triangleq \max_i d(i)$. 
Then, the set of all feasible schedules $\sI(G)$ is given by:
$$\sI(G)\triangleq \{\bsigma \in \{0,1\}^{n}:\sigma_i+\sigma_j \le 1, \forall (i,j) \in E \}.$$
The feasible rate region $C(G),$ which is the set of all possible
service rates over the links, is simply the convex hull of $\sI(G)$,
defined as follow: 
\sqeq
\begin{equation*}
C(G) \triangleq \\ \left\{ \sum_{\bsigma \in
    \sI(G)}\alpha_{\bsigma}\bsigma :\sum_{\bsigma \in
    \sI(G)}\alpha_{\bsigma}=1,~ \alpha_{\bsigma}\ge 0,~ \forall \bsigma
  \in \sI(G) \right\}.
\end{equation*}
\unsqeq

\smallskip
\noindent{\fbs CSMA (Carrier Sense Multiple Access).} 
Now we describe a CSMA algorithm which updates the scheduling vector $\bsigma(t)$
in a distributed fashion. Initially, the algorithm starts with the null schedule, i.e., $\bsigma(0)={\bf 0}$.
Each link $i$ maintains an independent Poisson clock of unit rate, and
when the clock of link $i$ ticks at time $t$, update its schedule $\sigma_i(t)$ as
\begin{itemize}
	\item[$\circ$] If the medium is sensed busy, i.e., 
	there exists $j\in \N(i)$ such that $\sigma_j(t)=1$, then $\sigma_i(t^+)=0$.
	\item[$\circ$] Else, $\sigma_j(t^+)=1$ with probability $\frac{\exp(r_i)}{\exp(r_i)+1}$
	and $\sigma_j(t)=0$ otherwise. 
\end{itemize}
In above, $r_i>0$ is called the {\em transmission intensity} (or simply
{\em intensity}) of link $i$.
The schedule $\sigma_i(t)$ of link $i$ remains unchanged while its clock does not tick.



Under the algorithm, the scheduling process $\{\bsigma(t):t\geq 0\}$
becomes a time reversible Markov process. It is easy to check that
its stationary distribution for given $\bor=[r_i]$ becomes:
\sqeq
\begin{equation}
\pi^{\bor}=[\pi^{\bor}_{\bsigma}: \bsigma \in
\set{I}(G)]~\mbox{where}~\pi^{\bor}_{\bsigma}=\frac{\exp\left(\sum_{i \in \slink}\sigma_ir_i\right)}{\sum\limits_{{\bm \rho} \in
\set{I}(G)}\exp\left(\sum_{i \in \slink}\rho_ir_i\right)}.
\label{eq:1}
\end{equation}
\unsqeq
In other words, the stationary distribution is expressed as
a product form of transmission intensities over links.
Then, due to the ergodicity of Markov process $\{\bsigma(t)\}$, the
long-term service rate of link $i$ is a function of transmission
intensity $\bor,$ which is the sum of all stationary
probabilities of the schedules where $i$ is active. We denote by
$s_i(\bor)$ the service rate of link $i$, which is 
\begin{equation}\label{eq:service}
s_i(\bor) =\mathop{\sum_{\bsigma \in \sI(G)}}_{:\sigma_i=1}\pi^{\bor}_{\bsigma} =\frac{\sum_{\bsigma \in \sI(G):\sigma_i=1}\exp(\sum_{i \in \slink}\sigma_ir_i)}{\sum_{\bsigma' \in \sI(G)}\exp(\sum_{i \in \slink}\sigma'_ir_i)}.\end{equation}

\subsection{Problem Description: P1 and P2}

In this section, we describe two central problems for designing CSMA
algorithms of high performances.  In a wireless network where CSMA is
used as the medium access control (MAC) mechanism, suppose packets arrive with rate $\lambda_i>0$ at
link $i$.  Then, the first-order question is about its stability, i.e.,
whether the total number of packets remains bounded as a function of
time.  Under the wireless network model considered in this paper, it is
not hard to check that the necessary and sufficient condition for
stability is that the service rate $s_i$ is larger than the arrival rate
$\lambda_i$. Therefore, this motivates the following question for the
CSMA algorithm design.

\medskip
\begin{compactenum}[\bf P1.]
	\item {\bf Stability.} For a given rate vector
  $\blambda=[\lambda_i] \in C(G)$, how can each link $i$ find its 
transmission intensity $r_i$ in a
  distributed manner so that 
$$\lambda_i
  = s_i(\bor),\qquad\mbox{for all links $i \in V$}?$$

\end{compactenum}
\smallskip
For the simple presentation of our results, we consider $\lambda_i
  = s_i(\bor)$ instead of $\lambda_i < s_i(\bor)$ in the description of the stability problem.
However, one can also obtain $\lambda_i < s_i(\bor)$ by solving {\bf P1}
with $\lambda_i \leftarrow \lambda_i+\varepsilon$ for small $\varepsilon>0$.

The second problem arising in wireless networks is controlling congestion, i.e, 
how to control the CSMA's intensity $\bor$ so that the
resulting rate allocation maximizes the total utility of the network.
Formally speaking, we study the following question.

\medskip
\begin{compactenum}[\bf P2.]
\item {\bf Utility Maximization.} Assume that each link $i$ has its utility function
$U_i:[0,1]\to \mathbb R_+$. How can each link $i$ find its
  transmission intensity $r_i$ in a
  distributed manner so that the total utility 
  $\sum_{i \in \slink}U_i (s_i(\bor))$ is maximized?
Our main optimizing goal is 
\begin{equation}
\text{\bf (OPT)} \qquad
\max_{\bor}\quad \sum_{i \in \slink} U_i (s_i(\bor)),\label{pro:util}
\end{equation}
when $U_i$ follows the class of $\alpha$-fair utility functions~\cite{mowa00}. 
\end{compactenum}
\section{Stability}
\label{sec:stability}

In this section, we present an approximation algorithm for the stability problem. The
problem finding a TDMA schedule (i.e., finding a repetitive scheduling
pattern over frames) to generate a target service rate vector has long
been studied, where the problem turns out to be NP-hard in many cases (a
variation of graph coloring) or allows polynomial time complexity only
for a special interference pattern such as node-exclusive interference,
see Chap. 2 of \cite{TDMA_thesis} for a survey. Even a distributed random access
based distributed algorithm requires exponentially long convergence
time in terms of the number of links \cite{yivesh06a}. 
The slow convergence of the prior CSMA-based iterative 
algorithms \cite{JSSW10DRA} for stability is primarily due to the fact
that it is hard to compute $s_i(\bor)$ given transmission intensity
$\bor$, i.e., it is not even clear whether the stability problem is in
NP. 

To overcome such a hurdle, we use a notion of free energy concepts in artificial intelligence and statistical
physics which allow to compute $s_i(\bor)$ efficiently in an approximate
manner. 


\subsection{Preliminaries: Free Energies for CSMA}\label{sec:freeenergyforcsma}

\noindent{\fbs Free energy functions.} 
We introduce the free energy functions for CSMA Markov processes
for transmission intensity $\bor$.
\begin{definition}[Gibbs and Bethe Free Energy]\label{def:energy}
$~$\\Given a random variable $\bsigma=[\sigma_i]$ on space $\sI(G)$ 
and its probability distribution $\nu$,
{\em Gibbs} free energy (GFE) and {\em Bethe} free energy (BFE) functions denoted
by $F_G(\nu;\bor)$ and $F_B(\nu;\bor)$ are defined as:
  \begin{align*}
    F_G(\nu;\bor)& = \E(\nu;\bor) - H_G(\nu), \quad
F_B(\nu;\bor) = \E(\nu;\bor) - H_B(\nu),
\end{align*}where $\E(\nu;\bor) = -\mathbb{E}[\bor \cdot \bsigma ]$, $H_G(\nu)=\mathbb{H}(\bsigma), $ and 
\begin{eqnarray*}
    H_B(\nu)&=& \sum_{i\in V} \mathbb{H}(\sigma_i) - \sum_{(i,j)\in E} \mathbb{I}(\sigma_i;\sigma_j).
\end{eqnarray*}
\end{definition}
In above, $\mathbb{E},$ $\mathbb{H},$ and $\mathbb{I}$ are the expected value, standard
entropy, and mutual information, respectively.  BFE can be thought
as an approximate function of GFE,\footnote{
  $F_B(\nu;\bor)=F_G(\nu;\bor)$ if the interference graph $G$ is a tree.} where $H_B$ is
called the `Bethe' entropy.
We note that in general the energy term $\E(\nu;\bor)$ can have a (different) form other than $-\mathbb{E}[\bor \cdot \bsigma ]$.

\smallskip
\noindent{\fbs How free energy meets CSMA.} 
The following theorem is a direct adaptation of the known results in literature (cf.~\cite{georgii1988gibbs}).
\begin{theorem}
  \label{thm:gibbs_maximal}
  The stationary distribution $\pi^{\bor}$ in~(\ref{eq:1}) of the CSMA Markov process with
  intensity $\bor$ is the unique minimizer of $F_G(\nu;\bor)$,
  i.e., $\pi^{\bor} = \arg\min_{\nu} F_G(\nu;\bor)$.
\end{theorem}

\smallskip
Theorem~\ref{thm:gibbs_maximal} provides a variational characterization
of $\pi^{\bor}$ (and thus the service rate vector $[s_i(\bor)]$).
Since BFE approximates GFE, the (non-rigorous) statistical physics
method suggests that a (local) minimizer or zero-gradient point (if exists) of
$F_B(\nu;\bor)$ can approximate $\pi^{\bor}$ (and $[s_i(\bor)]$).  The
main advantage of studying BFE (instead of GFE) is that
BFE depends only on the first-order marginal probabilities of joint distribution $\nu$,
i.e., its domain complexity is significantly smaller than that of GFE.

Specifically, by letting $y_i=\mathbb{E}[\sigma_i]$ and $\by=[y_i]$, which is
the service rate of link $i$, one can obtain the following expression:
\begin{eqnarray}
F_B (\nu;\bor) &=& -\sum_{i \in \slink} y_i r_i- \sum_{i \in
  \slink}\Big[(d(i)-1)(1-y_i)\log(1-y_i) -y_i\log y_i \Big] \cr 
  & & + \sum_{(i,j)\in E} (1-y_i-y_j)\log(1-y_i-y_j). \label{eq:bethe}
\end{eqnarray}
Namely, $F_B (\nu;\bor)$ is represented by service rate (or marginal probability) vector $\by$. Thus, without loss
of generality, we redefine BFE as a function of $\by$ as
following: $F_B(\by;\bor) = \E(\by;\bor)-H_B(\by),$ where
$\E(\by;\bor) = -\sum_{i \in \slink} y_i r_i$ and $H_B(\by)$ includes the other terms
in \eqref{eq:bethe}.
The
underlying domain $D_B$ of $F_B$ is
\begin{eqnarray}
  \label{eq:bethe_domain}
D_B~=~\{\by: y_i\geq 0, y_i+y_j\leq 1,~\mbox{for all}~(i,j)\in E\}. 
\end{eqnarray}

Hence, a (local) minimizer or zero gradient point $\by$ of
$F_B(\by;\bor)$ under the domain $D_B$ provides a candidate to approximate
$[s_i(\bor)]$, i.e., $y_i\approx s_i(\bor)$. It is known \cite{Yedidia05constructingfree} that the popular Belief Propagation (BP)
algorithm for estimating marginal distributions in MRFs can find the
zero gradient point $\by$ if it converges.
To summarize, the
advantage of studying BFE instead of GFE is that finding service rates
(or marginal distribution) reduces to solving a certain non-linear
system $\nabla F_B(\by;\bor)=0$ or $\nabla \Lambda(\by,\cdot)=0$,
where $\Lambda$ is the Lagrange function of $F_B(\by;\bor).$
Furthermore, one can prove that
there always exists a solution to $\nabla F_B(\by;\bor)=0$ using the
Brouwer fixed-point theorem.

In general, the service rates estimated by BFE do not coincide with
the exact service rates.  We formally define the error for this Bethe
approach as the maximum difference between the estimated rate and
the exact service rate across all links.
\begin{definition}[Bethe Error] \label{def:bethe}
  For a given transmission intensity $\bor,$ the Bethe
  error $e_B$ is defined 
  by: 
  $$e_B (\bor) = \max_{\by : \nabla F_B(\by;\bor)=0}\max_{i\in \slink} | y_i - s_i(\bor) |.$$
\end{definition}
It is not easy to bound the Bethe error for loopy graphs,
since it reduces to analyze the BP error. Despite the hardness of analyzing the BP error, 
BP often shows remarkably strong heuristic performance beyond tree-like graphs. This is
the main reason for the growing popularity of the BP
algorithm, and motivates our approach in this paper. Although there is
no known generic bound on the Bethe error for general graphs, one can prove that the Bethe error goes to 0 in the large-system limit, if the graph has no short cycle and its maximum degree is at most 5, 
i.e., sparse `tree-like' graph. 
For instance, the ring topology is
  an example of such graphs, the Bethe error over the ring
 topology goes to 0 as the number of nodes goes to infinity \cite{KL12ABP}.
The degree 5 condition is due to the known correlation decaying property \cite{chandrasekaran2011counting},
where quantifies the long range correlations in spin systems.



\subsection{BAS: Algorithm using Bethe Free Energy}\label{sec:bethe-approximation}

As discussed in Section \ref{sec:freeenergyforcsma}, an approximate solution to
the stability problem can be obtained by the Bethe free energy function: 
given a target service rate $s_i(\bor)$, s.t. $s_i(\bor) = \lambda_i,$
find the transmission intensity $\bor$ such that $\nabla
F_B(\blambda;\bor)=0$. Motivated by it, we propose the following algorithm: 
\vspace{0.2in}\hrule
\vspace{0.05in} \textbf{Bethe Algorithm for Stability: BAS($\blambda$)}
\vspace{0.05in} \hrule
\vspace{0.05in}

\begin{itemize} 
\item [$\circ$] Through message passing with neighbor links, each link
  $i$ knows $\lambda_j$ for all the neighbor links $j \in \n (i)$
\vspace{0.1in}
\item [$\circ$] Each link $i$ sets its transmission intensity $r_i$:
\begin{equation}\label{eq:b}r_i = \log \left(\frac{\lambda_i(1-\lambda_i)^{d(i) - 1}}{\prod_{j \in \N(i)} (1-\lambda_i - \lambda_j)}\right).
	\end{equation}
\end{itemize} 

\vspace{0.05in}
\hrule\vspace{0.2in}


In {\bf BAS}, a link sets its own transmission intensity based on the its
own and neighbors' arrival rates. With the closed form of equation
\eqref{eq:b}, each link can easily compute the transmission intensity
{\em without} any further iterations. We now state the main property of
{\bf BAS.}

\begin{theorem}
  \label{thm:bethe}
	For the choice of $\bor=[r_i]$ 
	by \eqref{eq:b}, it follows that
	$$\nabla F_B(\blambda;\bor)=0.$$	
\end{theorem}
From \eqref{eq:bethe}, it is trivial to prove
Theorem~\ref{thm:bethe}. It is noteworthy that the BFE
function with some $\bor=[r_i]$ may not have any local minima strictly
inside of its domain, which indicates that `estimation-and-update'
using BP or BFE even
may not converge at all whereas {\bf BAS} requires just one computation.

Since the Bethe free energy function does not give the exact
  solution except for tree graphs, $s_i(\bor )$ under {\bf BAS} might
  be less than $\lambda_i$ for some links $i$. To guarantee $s_i(\bor
  ) \ge \lambda_i$ for every link $i$, we can use conventional CSMA
  algorithms such as \cite{JSSW10DRA} and \cite{JW10DC} after {\bf
    BAS}. Since {\bf BAS} is
a sort of `offline' algorithms which does not need estimations on service
rates, BAS can choose `good' initial
transmission intensities for the conventional 
CSMA algorithms to boost up the convergence speeds of  
CSMA algorithms, while guaranteeing the maximal stability.

\section{Utility Maximization}
\label{sec:utilmax}

In this section, we present an approximation algorithm for the network
utility maximization problem \eqref{pro:util}.
To design a distributed algorithm finding transmission intensity
$\bor$ for \eqref{pro:util}, the approaches in literature
\cite{JSSW10DRA,IEEEpaper:Liu_Yi_Proutiere_Chiang_Poor_2009,HP12SO},
instead, considers the following variant of \eqref{pro:util}: for
$\beta>0$,
\begin{align}
 \max_{\bor} & \quad \beta \cdot\sum_{i \in \slink}
 U_i (s_i(\bor)) +H_G(\pi^{\bor}).\label{pro:util-1} 
\end{align}
The proposed algorithms \cite{JSSW10DRA,IEEEpaper:Liu_Yi_Proutiere_Chiang_Poor_2009,HP12SO}
converge to the solution to \eqref{pro:util-1}. Since
the entropy term $H_G(\pi^{\bor})$ is bounded above and below, the solution
to \eqref{pro:util-1} can provide an approximate solution to
\eqref{pro:util} if $\beta$ is large.

\subsection{BUM: Algorithm using Bethe Free Energy}
In BFE functions, the Bethe entropy $H_B(\by)$ is exploited
instead of the Gibbs
entropy $H_G(\pi^{\bor})$, which significantly reduces the complexity
to find a solution. As the BFE functions, we modify \eqref{pro:util-1} as follows:
\begin{align} \label{eq:uopro}
\max_{\y \in D_B} & \quad K_{B}(\by) = \beta \cdot \sum_{i \in \slink} U_i
(y_i) +H_B (\by)
\end{align}
where 
the Bethe entropy allows to replace the term $s_i(\bor)$ by a new variable $y_i,$ and
 the domain constraint $D_B$ given by \eqref{eq:bethe_domain} is necessary to evaluate $H_B(\by)$. 
Once \eqref{eq:uopro} is solved, 
one has to recover $\bor$ from $\by$ such that $s_i(\bor)=y_i$. 
To summarize, our algorithm for utility maximization consists of two phases:
\begin{itemize}
	\item[1.] Run a (distributed) gradient algorithm solving \eqref{eq:uopro} and obtain $\y$.
	\item[2.] Run the {\bf BAS} algorithm to find a transmission intensity $\bor$
	for the target service rate $\by$.
\end{itemize}
The algorithm is formally described in the following:

\vspace{0.2in}\hrule
\vspace{0.05in}
\textbf{Bethe Utility Maximization: BUM}
\vspace{0.05in}
\hrule
\vspace{0.05in}
\begin{itemize}

\item[$\circ$] 
Initially, set $t=1$ and $y_i(1) = 1/4, ~i\in \slink$. \footnote{The
  initial point can be any feasible point in $D_B.$ The point, $y_i =
  1/4$ for all $i$, is such a feasible point. }

\vspace{0.1in}

\item[$\circ$] \underline{\sl Intensity-update based on $\by$.}

 Obtain $(y_j, j \in
  \set{N}(i))$ through message passing with the neighbors, and set transmission intensity
  $r_i (t)$ of link $i$ for time $t$ :
\begin{equation}\label{eq:biu}
r_i (t) =\log \left(\frac{y_i (t) \big( 1-y_i (t) \big)^{d(i) - 1}}{\prod_{j
      \in \N(i)} (1-y_i (t) - y_j (t))}\right).\end{equation}

\vspace{0.1in}

\item[$\circ$] \underline{\sl $\by$-update based on time-varying gradient
    projection.}

 $y_i(t+1)$ is updated for time $t+1$ at each link $i$:
\begin{equation*}
y_i(t+1) =\left[y_i(t) + \frac{1}{\sqrt{t}}\frac{\partial K_B}{\partial
y_i}\Bigg|_{\by(t)} \right]_*,
\end{equation*}
where the projection $[\cdot]_*$ is defined as follows:
\begin{equation*}
[x]_{*}=
\begin{cases}
  c_1(t)& \text{if}~x<c_1(t)\cr
1- \kappa(t) & \text{if}~x> 1- \kappa(t)\cr
x&\text{otherwise}
\end{cases},
 \qquad\kappa(t)=\frac{1-y_i(t)+{\max_{j\in\n(i)}y_j(t)}+c_2(t)}2.
\end{equation*}

\end{itemize}
\hrule\vspace{0.2in}

\noindent{\fbs $\by$-update.} In the $\by$-update phase, each link $i$
updates $y_i$ in a distributed manner based on a gradient-projection
method. However, our projection $[\cdot]_{*}$ is far from a classical
projection, where our projection varies over time (see $c_1(t)$ and
$c_2(t)$), which our algorithm's convergence and distributed operation
critically relies on. 
We delay the discussion on why and how our
special projection contributes to the theoretical performance
guarantee of {\bf BUM}, and first present its feasibility of
distributed operation. Note that the gradient $\frac{\partial
  K_B}{\partial y_i}$ in the $\by$-update phase is:
\begin{align}
  \frac{\partial K_B}{\partial y_i}\Bigg|_{\by(t)}& ~= ~ \beta\cdot
  U_i'(y_i(t))-(d(i)-1)\log(1-y_i(t))\cr &-\log y_i(t)+\sum_{j \in
    \n(i)}\log(1-y_i(t)-y_j(t)),\label{eq:bethegradient}
\end{align}
Indeed, this gradient can be easily obtained by the link $i$ via local
message passing only with its neighbors. Since 
$\by(t)$ has to be an interior point of $D_B$ for computing the gradient
  \eqref{eq:bethegradient}, a projection is necessary in {\bf BUM}.

\smallskip
\noindent{\fbs Performance guarantee.} 
We now establish the theoretical performance guarantee of {\bf BUM} for
the popular class of $\alpha$-fair utility functions \cite{mowa00}, i.e.,
$$U_i(x) = \begin{cases}
\log x&\mbox{if}~\alpha =1\\
\frac{x^{1-\alpha}}{1-\alpha}&\mbox{otherwise}\end{cases}.$$ The parameter $\alpha$ represents
the degree of fairness for the throughput allocation: when $\alpha = 0,$
the total link throughput is maximized; $\alpha=1$ gives the
proportional fair allocation when $\alpha \rightarrow \infty,$ it
corresponds to the max-min fairness.

Let $\by^*$ be an optimum point of $K_B$, i.e., 
$\by^* = \arg \max\limits_{\by \in D_B}K_B(\by).$ 
The following theorem shows that, for any given $\alpha,$ with
sufficiently large $\beta,$ $K_B(\by(t))$ by {\bf BUM} always
converges to $K_B(\by^*)$ in polynomially large enough time $T$ with resepct to $n$.
\begin{theorem} \label{thm:bethe-con} Let $\mu$ be a probability
  distribution on $\{1,\ldots, T\},$ such that
$$\mu(t)=\frac{t^{-1/2}}{\sum_{s=1}^{T} s^{-1/2}}\qquad\mbox{for}~ t\in \{1,\dots, T\}.$$
Then, if  $\beta > 2d / \alpha$, 
\begin{equation}\lim_{t\to \infty} \max\{c_1(t), c_2(t) \} = 0 \quad\mbox{and}\quad\lim_{t\to \infty}\frac{1}{\sqrt{t}}\frac{(c_1(t))^{-\alpha} - \log
  \big( c_1(t)c_2(t) \big)}{\min\{c_1(t),c_2(t)\}} = 0,\label{eq:c1c2}\end{equation}
it follows that
\begin{equation}
\label{eq:bound}
\mathbb{E}\left[K_B(\by^*)- K_B(\by(t)) \right] \leq 
O\left(\frac{n \log T}{\sqrt{T}}\right),
\end{equation}
where the expectations are taken over the distribution $\mu$.
\end{theorem}
The proof of the above theorem is given in Section~\ref{sec:pfthm:bethe-con}.
Our key intuition underlying the proof is that 
the projection $[\cdot]_*$ of {\bf BUM}
is designed so that 
the
updating $\boldsymbol{y}(t)$ never hits the projection boundary of $D_B$ after
a time instance $t^{*}$. Then, one can observe that 
the
algorithm behaves as a gradient algorithm without a projection
after time $t^{*}$, and hence it is possible to analyze its convergence using traditional techniques.
We note that for $\beta > 2d/ \alpha,$ $\by(t)$ always
converges to the unique $\by^*$ when \eqref{eq:c1c2} holds, since
$K_B$ is a (strictly) concave function. There exist many paris of
$(c_1(t),c_2(t))$ satisfying \eqref{eq:c1c2}, \eg,~
$$c_1(t) = - C_1 \log c_2(t),\qquad c_2(t) =C_2 t ^{-\gamma},$$ where $C_1$
and $C_2$ are some constants and $0<\gamma < 1/2$. 
The following theorem further bounds the gap between the achieved utility of
{\bf BUM} and the maximum utility. 
\begin{theorem} \label{thm:bethe-error} The transmission
  intensity $$\bor^*:=\left[\log \left(\frac{y_i^* \big( 1-y_i^*
        \big)^{d(i) - 1}}{\prod_{j \in \N(i)} (1-y_i^* -
        y_j^*)}\right)\right]$$ satisfies
\begin{equation*}
\label{errorofBUM}
\max_{\bx \in C(G)} \sum_{i \in \slink}U_i(x_i) - \sum_{i \in \slink}U_i\Big(s_i(\bor^{*})\Big) \leq 
\sum_{i \in \slink} \frac{e_B(\bor^{*})}{s_i
(\bor^{*})^{\alpha}} +\frac{n\log 2}{\beta}.
\end{equation*}
\end{theorem}
The proof of the above theorem is given in Section~\ref{sec:proof-theor-refthm:b}.
We recall that $e_B(\bor^*)$ is the Bethe error with
  transmission intensities $\bor^*$ which is defined in Definition~\ref{def:bethe}. 
As we mentioned earlier, the Bethe error $e_B(\bor^*)$ is small\footnote{In particular, $e_B(\bor^*)=0$ 
if the interference graph is a tree.} empirically 
in many applications \cite{Yedidia05constructingfree,F01CG,MWJ99LB}, and then 
the remaining error term is negligible for large $\beta$.

\subsection{Comparison with Prior Approach}\label{sec:comp-with-prior}

In \cite{JSSW10DRA,IEEEpaper:Liu_Yi_Proutiere_Chiang_Poor_2009},
gradient based algorithms solve \eqref{pro:util-1}. In this section,
we denote by {\bf JW} and {\bf EJW} (the names are used in
\cite{HP12SO}) the algorithms in \cite{JSSW10DRA} and
\cite{IEEEpaper:Liu_Yi_Proutiere_Chiang_Poor_2009},
respectively. Technically, the algorithms take the dual problem of
\eqref{pro:util-1} where transmission intensity $r_i$ is Lagrangian
multiplier and $U'^{-1} \Big(\frac{r_i(t)}{\beta} \Big) -
{s}_i(\bor(t))$ is the gradient of the dual problem \eqref{pro:util-1} for
$r_i$. Thus, transmission intensities are commonly described as the
following distributed iterative procedure:
\begin{eqnarray}
  \label{convex-csma-update}
  r_i(t+1) = r_i(t) + \alpha_i(t)\left ( U'^{-1}
    \Big(\frac{r_i(t)}{\beta} \Big)- s_i(\bor(t)) \right ),
\end{eqnarray}
where $\alpha_i(t)>0$ is the step size of link $i$. In both
schemes, $\alpha_i(t) = 1/t$ which guarantees the convergence of
$r_i(t)$. However,
to update $r_i(t+1)$ as per \eqref{convex-csma-update}, $s_i(\bor(t))$ is
hard to compute. For the issue,
a empirical service rate $\hat{s}_i(t)$ has been used instead of
${s}_i(\bor(t))$. 

The authors in \cite{JSSW10DRA} take a
large and increasing length of intervals (i.e., $r_i(t)$ is fixed during each
interval) so that ${s}_i(\bor(t))$ can be estimated well by its empirical estimation $\hat{s}_i(t)$ at the end of
each interval. On the other hand, the authors in
\cite{IEEEpaper:Liu_Yi_Proutiere_Chiang_Poor_2009}, with a fixed
length of intervals (which does not have to be very large), use the
empirical estimation $\hat{s}_i(t).$ By stochastic approximation, with
sufficiently large $T,$
$$\lim_{t \rightarrow \infty}r_i(t+T) -r_i(t) =  \sum_{j=t}^{t+T} \alpha(j)\left( U'^{-1}
    \Big(\frac{r_i(j)}{\beta} \Big)-
s_i(\bor(j)) \right).$$  
Both approaches, however, suffer from slow convergence: the updating interval should be extremely large
in \cite{JSSW10DRA} and $\alpha_i(t)$ should be extremely small in
\cite{IEEEpaper:Liu_Yi_Proutiere_Chiang_Poor_2009} for $\hat{s}_i(t) \approx
{s}_i(\bor(t))$.

In \cite{HP12SO}, the authors propose an algorithm called Simulated
Steepest Coordinate Ascent ({\bf SSCA}) algorithm converging to the same point
with the above two algorithms, where the algorithm is not a gradient based
approach but a steepest ascent based algorithm. In SSCA scheme, at each
iteration $t$, link $i$ sets transmission intensity as $r_i (t) =
\beta U' ( \frac{1}{t}\sum_{j=1}^{t} \hat{s}_i(j)).$ Then,
$\pi^{\br}_{\bsigma}$ is maximized at
$\sigma^* := \arg
\max_{\sigma \in \sI(G)} \sum_{i \in \slink} \sigma_i
U'\left(\frac{1}{t}\sum_{j=1}^{t} \hat{s}_i(j)\right),$ 
which is the
exact steepest ascent direction.  As the steepest ascent algorithms
converge to the optimal service rates in many applications, the SSCA
algorithm makes the service rates converge to the optimal rates
quickly, compared to the gradient based algorithms.  To guarantee the
convergence, however, SSCA algorithm may still have to spend extremely large
iterations since schedules are stochastically selected over time.

\subsection{Proof of Theorem \ref{thm:bethe-con}}\label{sec:pfthm:bethe-con}

We first give an overview for the proof of Theorem \ref{thm:bethe-con}.
The formal complete proof will follow.

\smallskip
\noindent{\fbs Overview of the proof of Theorem \ref{thm:bethe-con}.}
We first prove that the function $K_B$ is concave for large enough $\beta$,
stated as follows.

\begin{lemma}
When $\beta \ge 2d/\alpha,$ $K_B (\by)$ is concave.\label{LEM:CONCAVE}
\end{lemma}
\begin{proof}
The proof  is presented in Appendix.
\end{proof}
We note that $K_B$ is not obvious to be concave (or convex) since the Bethe entropy term $H_B$
(in the expression of $K_B$)
is neither concave nor convex. In essence, we observe that the non-concave term $H_B$
is compensated by the concave term $\beta \cdot \sum_{i \in \slink} U_i
(y_i)$ for large enough $\beta$. 

The concavity property of $K_B$ might allow to use known convex optimization tools such as the interior-point method, the Newton's
method, the ellipsoid method, etc.  However, these algorithms are not
easy to implement in a distributed manner, and it is still far from
being clear whether a simple distributed gradient algorithm can solve \eqref{eq:uopro} (in a polynomial number
of iterations) since the optimization is `constrained', i.e., $y_i\geq
0$ and $y_i+y_j\leq 1$ for $(i,j)\in E$. Thus, we carefully design the
dynamic projection $[\cdot]_{*},$ where $c_1(t)$ and $c_2(t)$ enforce
$\by(t)$ to be strictly inside of $D_B.$ Lemma~\ref{lem:betheupdate}
is the key lemma of this proof, where
we show that 
$c_1(t)<y_i(t)<1-c_2(t)$ after large enough $t$. Since
the algorithm does not hit the `boundary' of $[\cdot]_{*}$ anymore
after large enough updates, {\bf BUM} acts like a gradient algorithm in `unconstrained' optimization.
\begin{lemma} For all time $t,$
 $\by(t)=[y_v(t)]\in D^*_B$, where
\begin{align*}D^*_B:=\{\by=[y_v]:y_v\in[\delta_1,1-\delta_2]&~~\mbox{and}~~  y_u+y_v\leq 1-\delta_3, \mbox{for all}~ 
 (u,v)\in E\}.\end{align*}
where
\begin{align*}
\delta_2&:=\min\left \{ c_2(t_*), \frac{1}{2 (\exp(\beta2^{\alpha})+1)}\right \},
\quad \delta_1:=\min  \left\{c_1(t_*),\frac{\beta2^{\alpha}\delta_2^d}{4(1+\beta2^{\alpha}d\delta_2^{d-1})}
  \right\}, \cr
\delta_3&:=\min\left \{ c_2(t_*), \frac{\delta_1}{2
    \exp(\beta \delta_1^{-\alpha})} \right\},\end{align*}
and $$t_* := \inf \left\{ \tau :  \frac{1}{\sqrt{t}}\left|\frac{\partial K_B (\by(t))}{\partial
    y_i}\right| < \frac{1}{2}\min\{c_1(t),c_2(t)\} ~~\forall~ t \geq\tau\right\}.$$
\label{lem:betheupdate}
\end{lemma}
\begin{proof}
The proof  is presented in Appendix.
\end{proof}

\renewcommand{\y}{\by}

\smallskip
\noindent{\fbs Completing the proof of Theorem \ref{thm:bethe-con}.}
Now we proceed toward completing the proof of Theorem \ref{thm:bethe-con}.

First, from $\delta_1,$ $\delta_2,$ and $\delta_3$ in
Lemma~\ref{lem:betheupdate}, we define $\delta$ and $t_{\delta}$ as following: 
\begin{eqnarray*}
 t_{\delta} &:=&
\max\left\{c_1^{-1}(\delta_1),c_2^{-1}(\min
\{ \delta_2, \delta_3 \})\right\}.
\end{eqnarray*} Then, 
Lemma~\ref{lem:betheupdate} implies that for every time $t \ge t_{\delta},$
$$y_i(t+1) = y_v(t)+  \frac{1}{\sqrt{t}}\frac{\partial K_B (\y (t))}{\partial
y_i} .$$
Namely, the projection $[\cdot]_*$ is not necessary after time $t_\delta$. 
Thus, it follows that for $t>t_{\delta}$,
\begin{align*}
\| \by (t+1) - \by^* \|_{2}^{2} =& \| \by (t) + \frac{1}{\sqrt{t}}\nabla
K_B(\by(t))- \by^* \|_{2}^{2} \cr
 =& \| \by (t) - \by^{*}\|_{2}^{2} + \frac{1}{t}\| \nabla K_B(\by(t)) \|_{2}^{2}+ 2 \frac{1}{\sqrt{t}}\nabla K_B(\by(t))^{T}\cdot (\by (t) - \by^{*})
\cr
\stackrel{(a)}{\leq}& \| \by (t) - \by^{*}\|_{2}^{2}+ \frac{1}{t}\| \nabla K_B(\by(t)) \|_{2}^{2}+ 2 \frac{1}{\sqrt{t}} (K_B (\by (t)) - K_B (\by^{*})),
\end{align*}
where $(a)$ comes from the concavity of $K_B (\y)$ in Lemma \ref{LEM:CONCAVE}. 
By rearranging terms in the above inequality, we have
\begin{equation}
\frac{1}{\sqrt{t}} (K_B (\by^{*}) - K_B (\by (t)))\leq \frac{1}{2} \Big(\| \by (t) -
\by^{*}\|_{2}^{2} - \| \by (t+1) - \by^* \|_{2}^{2}  + \frac{1}{t}\| \nabla K_B(\by(t)) \|_{2}^{2} \Big). \label{eq:sumdiffer}
\end{equation}

We are now ready to complete this proof. We divide
$\sum_{t=1}^T \mu(t) ( K_B (\y^*) - K_B (\y(t)))$ into two parts: 
\sqeq
\begin{equation*}
\sum_{t=1}^T \mu(t) ( K_B (\y^*) - K_B (\y(t)))= 
\sum_{t=1}^{t_{\delta}-1} \mu(t) ( K_B (\y^*) - K_B (\y(t)))+
\sum_{t=t_{\delta}}^T \mu(t) ( K_B (\y^*) - K_B (\y(t))),
\end{equation*}
\unsqeq

\noindent where the first part can be bounded by some constant. We also obtain the upper bound
of the second part by \eqref{eq:sumdiffer}.
\sqeq
\begin{align*}
\sum_{t=t_{\delta}}^T \mu(t) ( K_B (\y^*) - K_B (\y(t)))~\leq~& \frac1{2\sum_{t=1}^{T}\frac{1}{\sqrt{t}}} \Big( \| \by (0) -
\by^{*}\|_{2}^{2} - \| \by (T) - \by^* \|_{2}^{2}  +  \sum_{t=t_{\delta}}^{T}\frac{1}{t}\,\|\nabla K_B
  (\y(t))\|_2^2 \Big) \cr
~\leq~&\frac1{\sum_{t=t_{\delta}}^{T}\frac{1}{\sqrt{t}}}\left(O(n)+ O\left(n\right)\,\sum_{t=t_{\delta}}^{T}\frac{1}{t}\right)
\end{align*}
\unsqeq
Finally, we can conclude that
\begin{equation*}
\sum_{t=0}^T \mu(t) ( K_B (\y^*) - K_B (\y(t)))=O\left(\frac{n \log T}{\sqrt{T}}\right).
\end{equation*}

\subsection{Proof of Theorem~\ref{thm:bethe-error}}\label{sec:proof-theor-refthm:b}
  There are two reasons for the error: the additional term of entropy
  in $K_B(\by)$ and the Bethe error because of intensity updating by
  \eqref{eq:biu}.  Thus, we devide the utility gap between the optimal
  value and the achieved value to represent the error due to each
  reason.
\begin{align*}
\max_{\bs \in C(G)}\sum_{i \in \slink} U_i (s_i) - \sum_{j \in
  \slink} U_i (s_i(\bor^*)) =&\left( \max_{\bs \in C(G)}\sum_{i \in \slink} U_i (s_i) - \sum_{j
    \in \slink}U_i (y_i^*) \right) + \cr &\left(\sum_{j
    \in \slink}U_i (y_i^*) -\sum_{j \in
  \slink} U_i (s_i(\bor^*)) \right)\cr
 \stackrel{(b)}{\leq}& \frac{H_B(\by^*)}{\beta} +\left(\sum_{j
    \in \slink}U_i (y_i^*) -\sum_{j \in
  \slink} U_i (s_i(\bor^*)) \right) \cr
 \stackrel{(c)}{\leq}& \frac{n\log 2}{\beta} + \sum_{i \in \slink} U_i
\big(s_i(\bor^*) +
e_B(\bor^*)\big) - \sum_{j \in
  \slink} U_i (s_i(\bor^*)) \cr
 \stackrel{(d)}{\leq}& e_B(\bor^*)\sum_{i \in
  \slink} s_i(\bor^*)^{-\alpha}+\frac{n\log 2}{\beta},
\end{align*}
where for $(b)$ we use $\beta \sum_{i \in \slink} U_i (s_i^*)
\leq K_B (\y^*)$, for $(c)$ we use the definition of Bethe error $e_B
(\bor^*)$ and $H_B (\by)
\leq n \log 2$, and $(d)$ holds since $U_i (\cdot )$ is an $\alpha$
fairness function and concave. This is the end of this proof.




\section{Simulation Results}
\label{sec:simulation}

In this section, we provide simulation results to verify how our
proposed algorithms perform under various scenarios. First, we compute
the Bethe error $e_B(\bor)$ (i.e., the difference between the target service rate and the
actual service rate) for various interference graphs and target service
rates. The tested interference graphs are shown in Fig.~\ref{fig:tps}.
Second, {\bf BUM} are compared with the three conventional
algorithms introduced in Section~\ref{sec:comp-with-prior} regarding to
convergence speed and achieved network utility, where we choose
$\alpha=1$, $\beta = 1$, $c_1(t) = \frac{1}{100 \log (t+e)}$, and $c_2(t) =\frac{1}{5t^{1/4}}$ just for simplicity.  We observed that other
values of $\alpha$ and $\beta$ show similar results.

\subsection{Stability}
As we stated in Section~\ref{sec:stability}, the stability algorithm
{\bf BAS} does not lead to the exact target service rate for the
topologies that are not tree. 
Fig.~\ref{fig:bethe_error} represents the Bethe
error for {\em complete}, {\em ring}, and {\em
  random} topologies. In the graphs, we define ``Load'' as the fraction of the
traffic rate over the capacity of the network and the $y$-axis
represents the normalized Bethe error by the target service rate. In
this experiment, we assume symmetric arrivals where the target service
rates of all links are equal.

\begin{figure*}[t!]
  \begin{center}
    \subfigure[\small Complete]{
      \includegraphics*[width = 0.2\columnwidth]{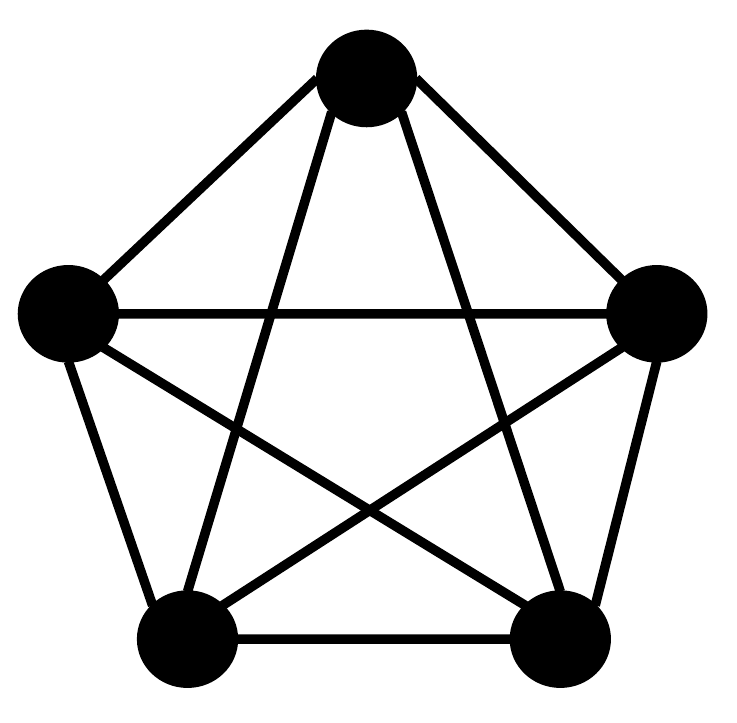}\label{fig:comp_tp} }
    \subfigure[\small Ring]{
      \includegraphics*[width = 0.2\columnwidth]{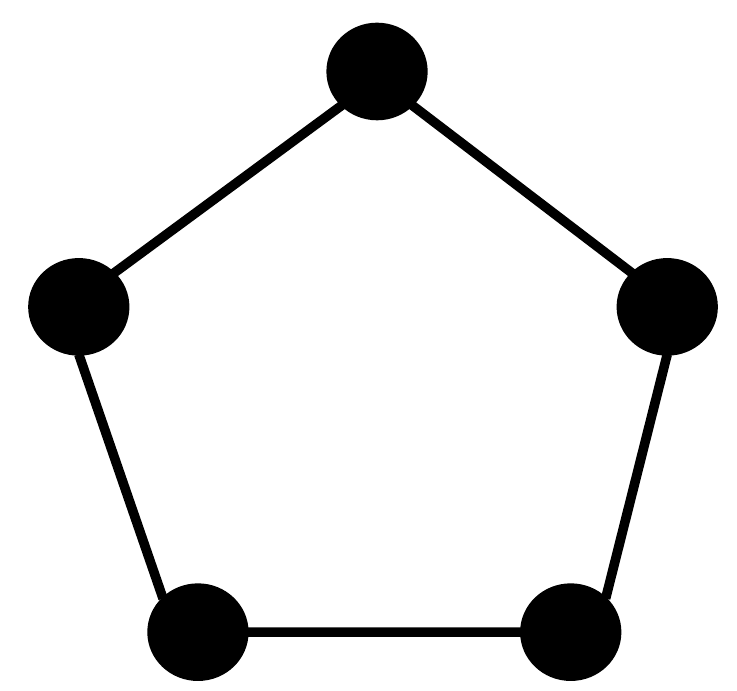}\label{fig:ring_tp}}
    \subfigure[\small Tree]{
      \includegraphics*[width =
      0.2\columnwidth]{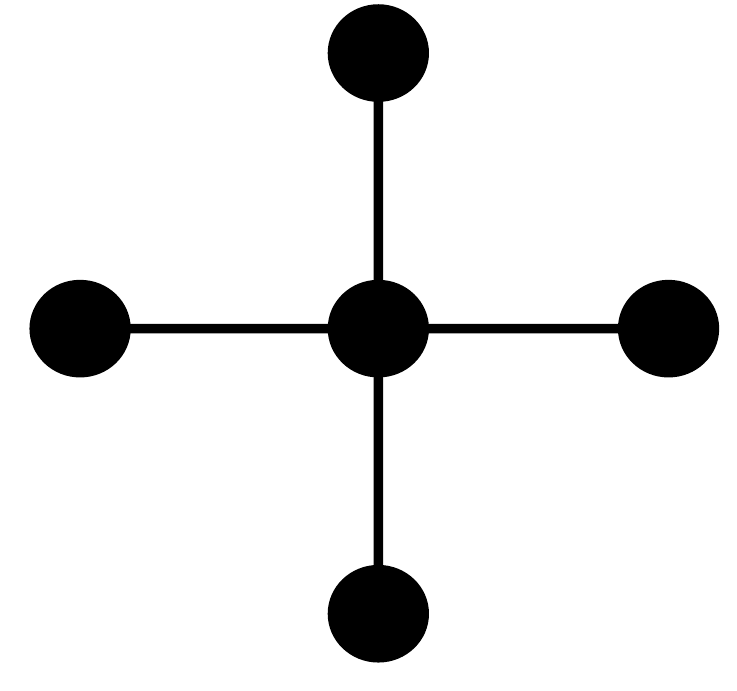}\label{fig:Tree_tp} }
    \subfigure[\small Grid]{
      \includegraphics*[width = 0.19\columnwidth]{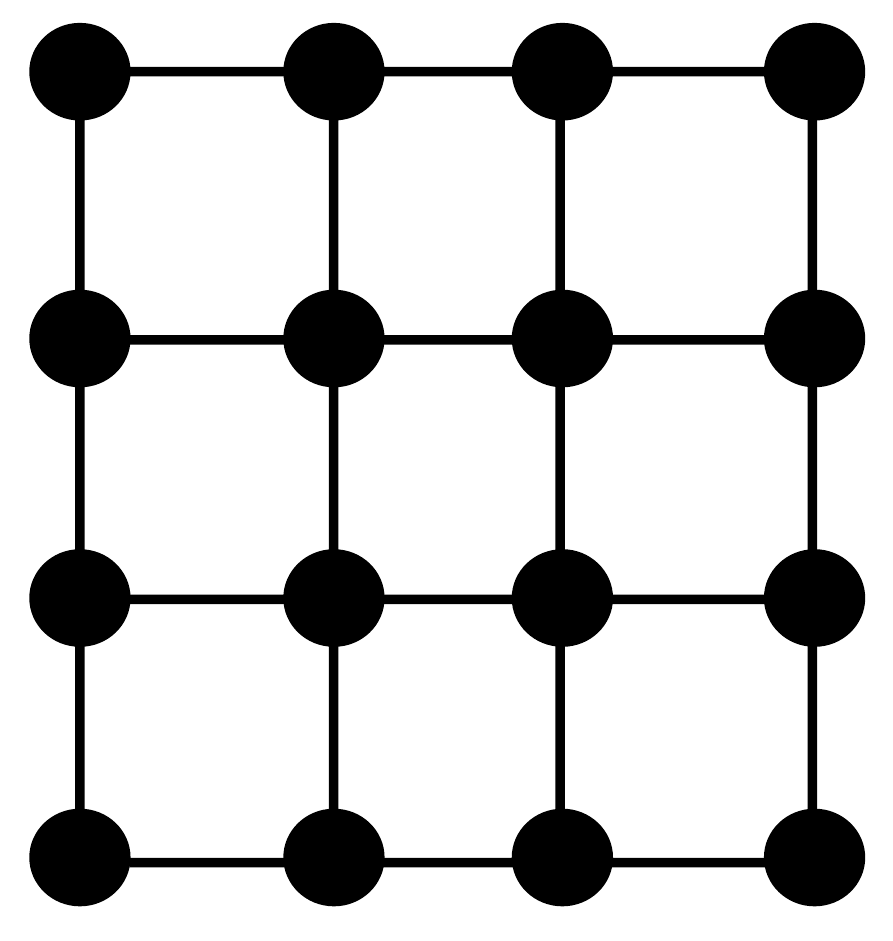}\label{fig:Tree_tp} }
  \vspace{-0.2cm}
    \caption{Interference graph topologies}
    \label{fig:tps}
  \end{center}
  \vspace{-0.2cm}
\end{figure*}

\smallskip
\noindent{\fbs Varying traffic loads.} The graphs in
Fig.~\ref{fig:bethe_error} show the normalized Bethe error on three
topologies: {\em complete,} {\em ring}, and {\em grid.} The normalized
Bethe errors grow up to at most 0.2, which means that the Bethe error is
within 20\% of the corresponding target service rate. In addition, for
all three topologies, the Bethe error increases as the traffic load
increases. Although {\bf BAS} experiences more errors with higher
transmission intensity, it is noteworthy that the mixing time also
increases with higher transmission intensity. Thus, the MCMC based
algorithms need far more convergence time with the higher transmission
intensity although they can get the accurate
service rate estimation.

\smallskip
\noindent{\fbs Impact of topology. } Bethe error should strongly
depend on the underlying topology. As stated in Section~\ref{sec:stability},
tree topologies do not have error, while other types of topologies have
positive Bethe error. Trees are the ones that are connected and have no cycle. In
general, cycles are the major reasons for large Bethe errors, where
errors tend to grow with the increasing number of cycles in the
topology. 
In this context, we observe that for complete graphs, the error
becomes more significant as the number of links increases, mainly
because the number of cycles also increases with the number of links.
For ring graphs, we also see the effect of the size of cycle. In
Fig.~\ref{fig:berr_ring}, the error of 12-links is smaller than that of
others, because the cycle becomes similar with a line topology as the
number of links increases.

\begin{figure*}[t!]
  \begin{center}
    \subfigure[\small Complete]{
      \includegraphics*[width = 0.31\columnwidth]{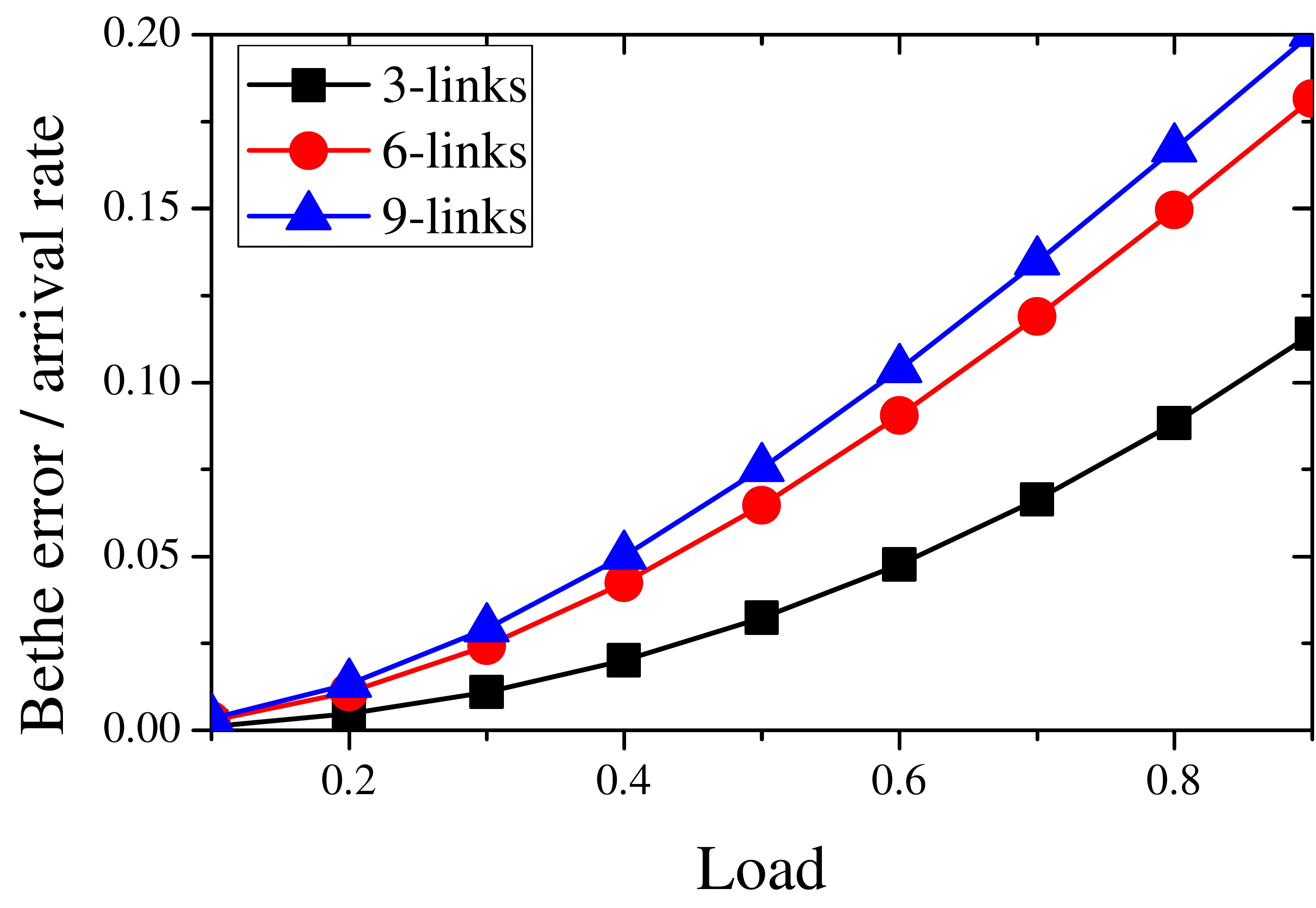}\label{fig:berr_comp} }
    \subfigure[\small Ring]{
      \includegraphics*[width = 0.31\columnwidth]{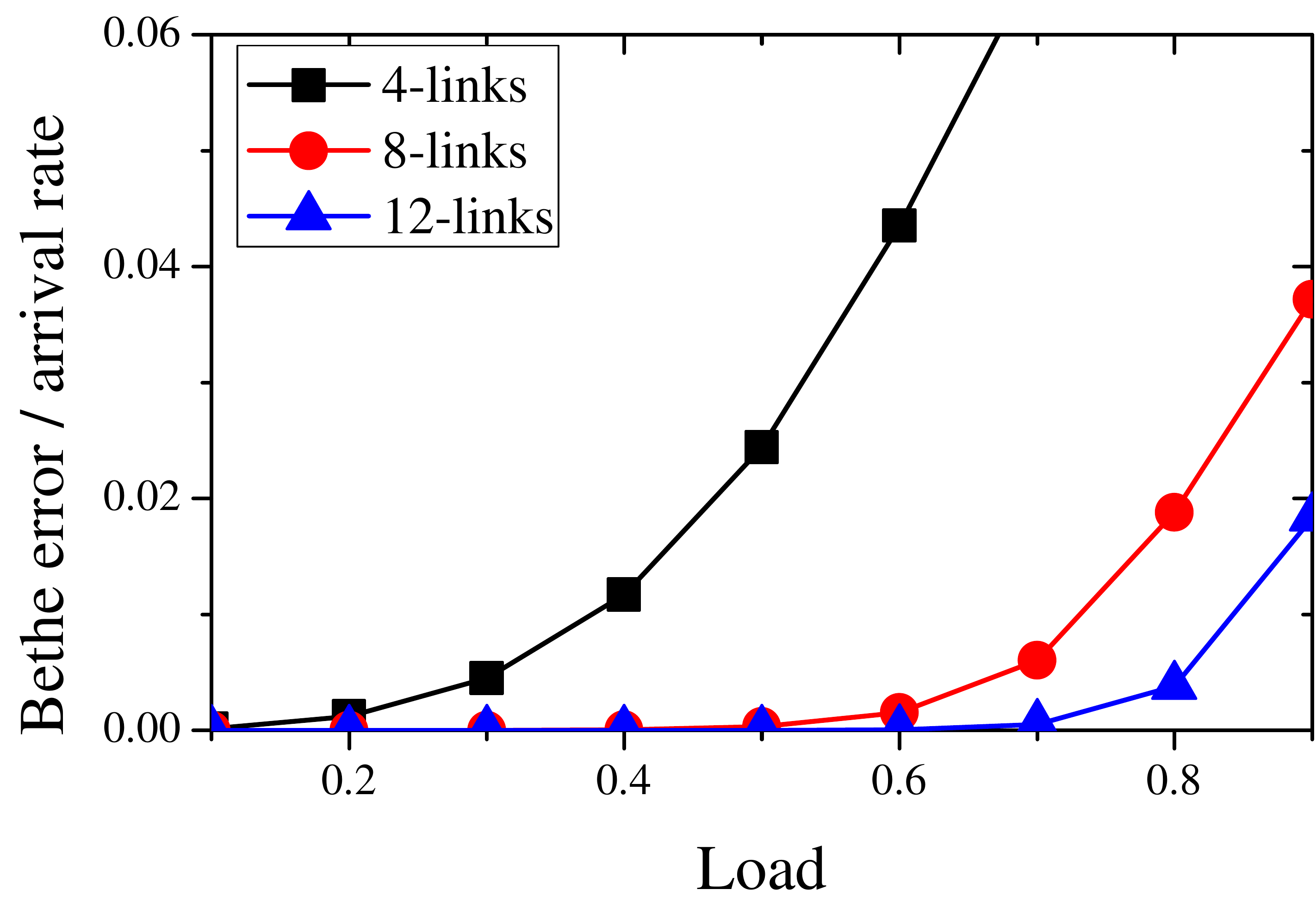}\label{fig:berr_ring}}
    \subfigure[\small Random]{
      \includegraphics*[width = 0.31\columnwidth]{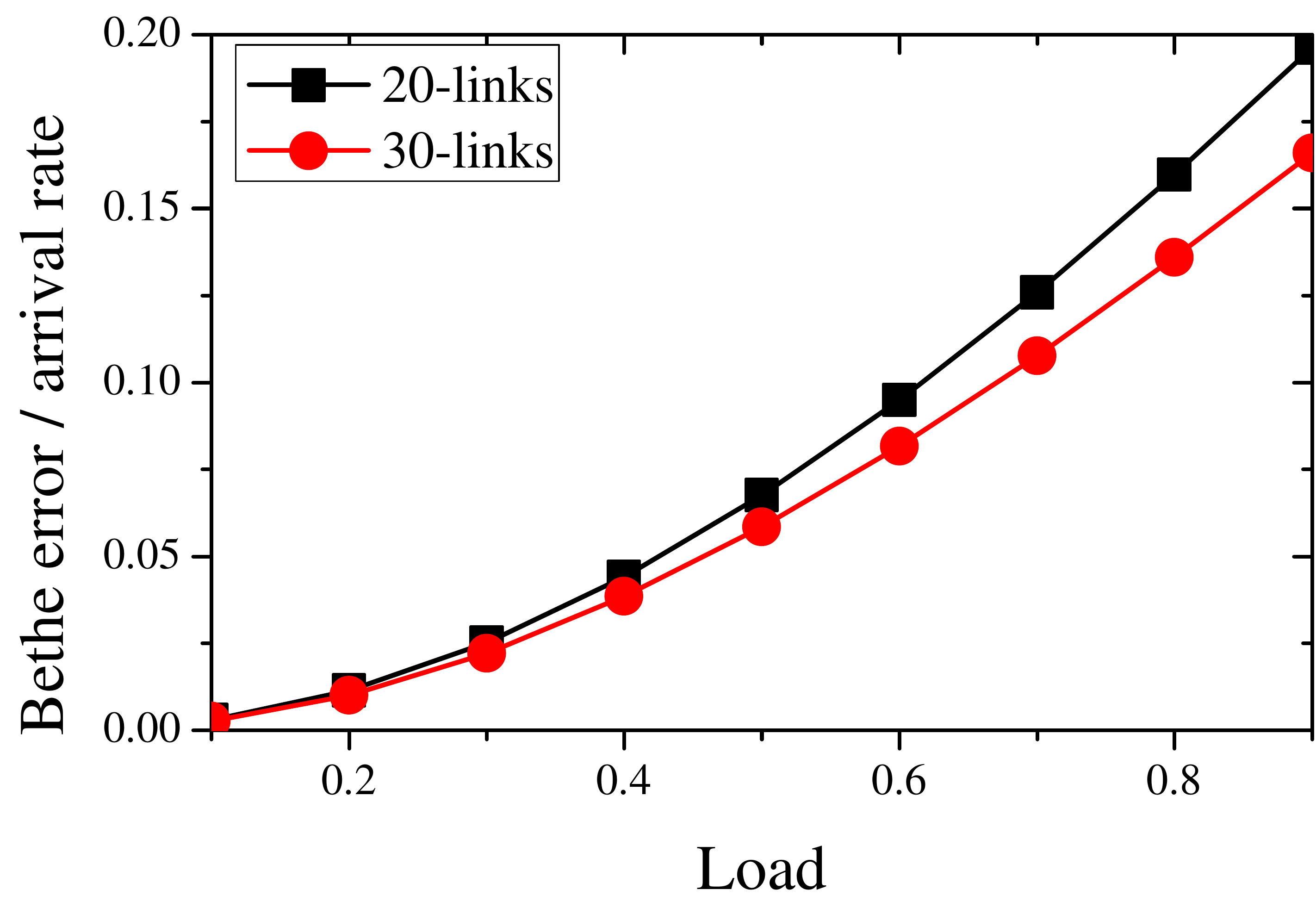}\label{fig:berr_rand} }
  \vspace{-0.2cm}
    \caption{Bethe error for various graphs (where `Load' means arrival
    rate / capacity when arrival rates are the same) }
    \label{fig:bethe_error}
  \end{center}
  \vspace{-0.2cm}
\end{figure*}

\begin{figure*}[t!]
  \begin{center}
    \subfigure[\small Tree (Star)]{
      \includegraphics*[width = 0.31\columnwidth]{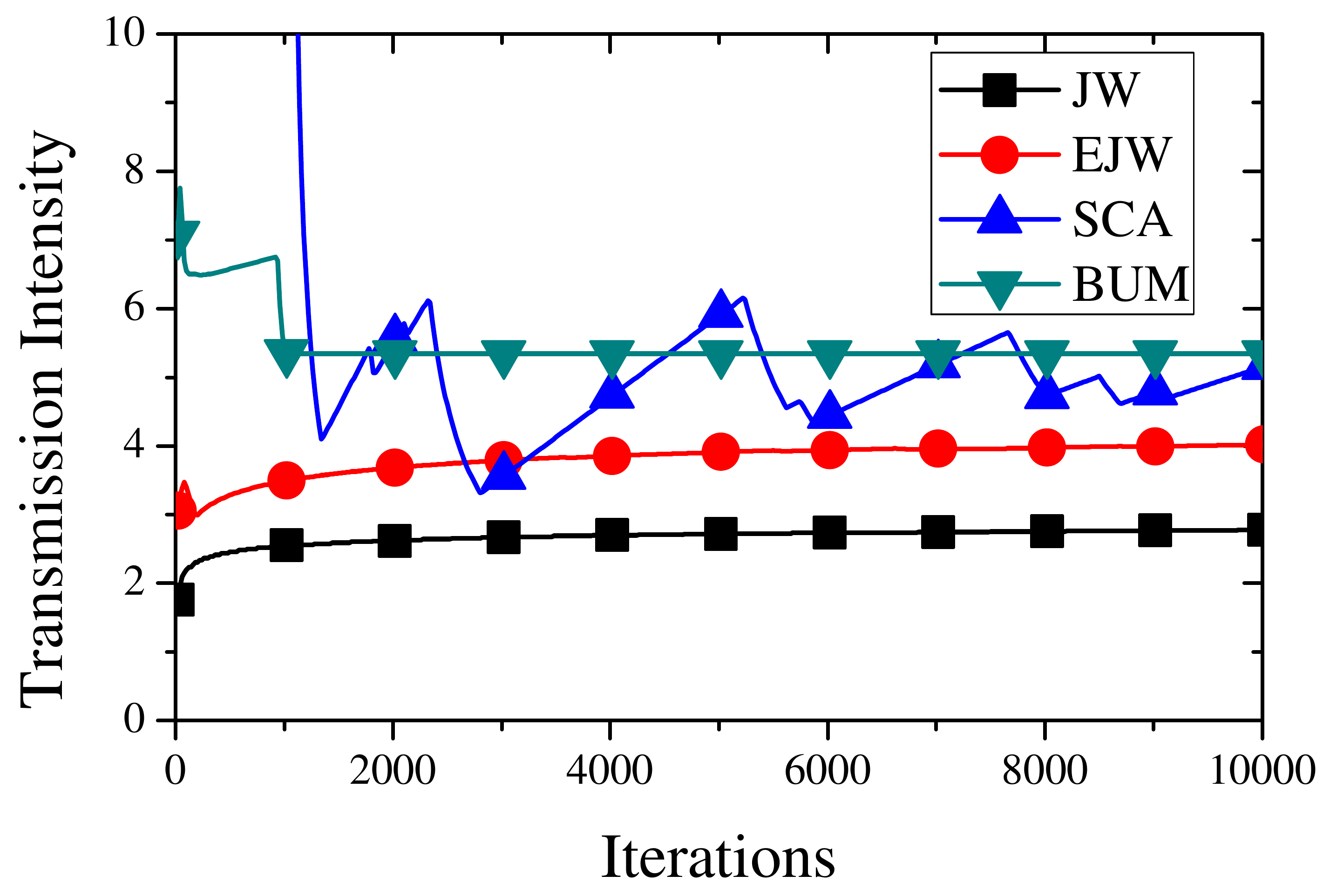}\label{fig:tr_tree} }
    \subfigure[\small Complete]{
      \includegraphics*[width = 0.31\columnwidth]{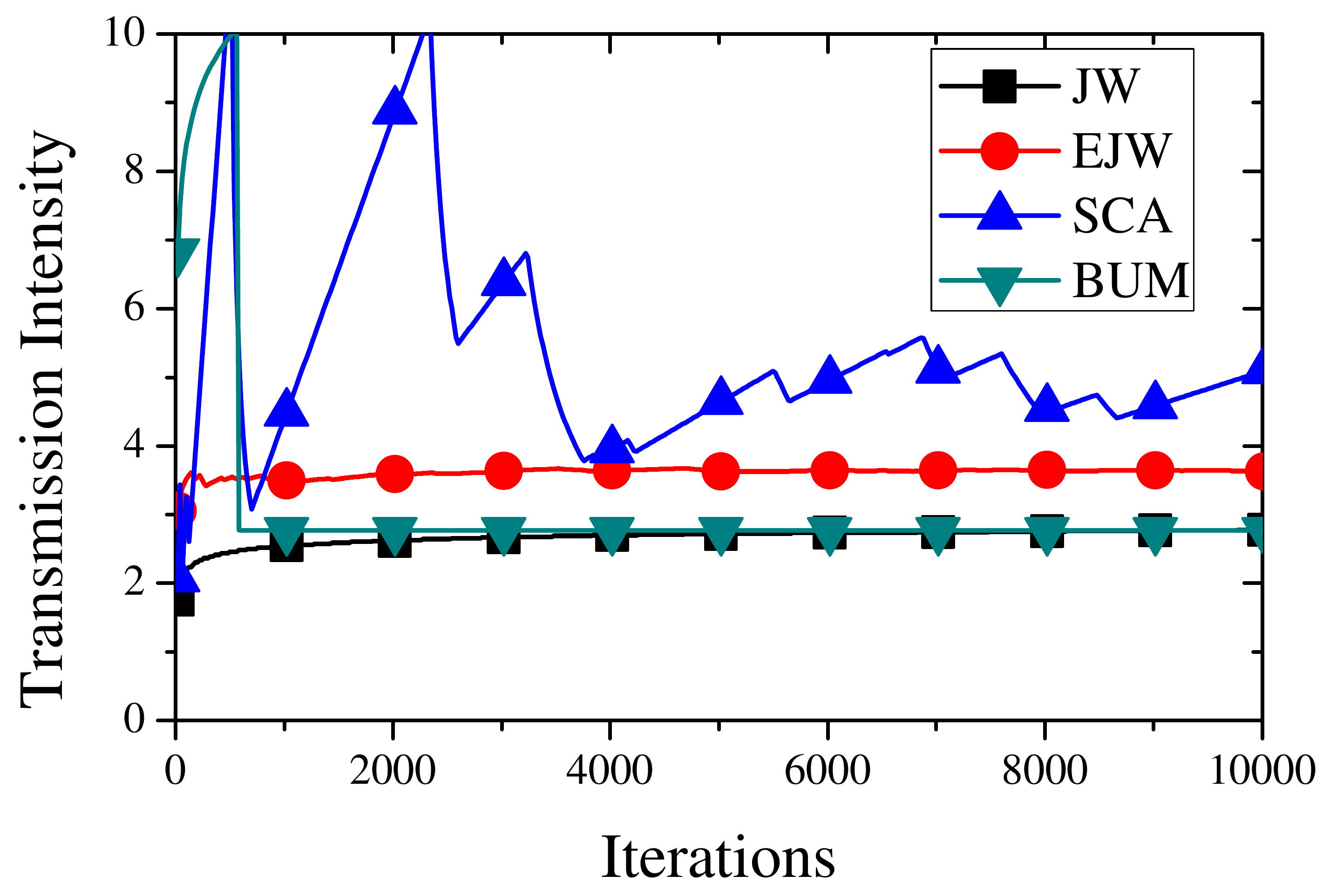}\label{fig:tr_comp}}
    \subfigure[\small Grid]{
      \includegraphics*[width = 0.31\columnwidth]{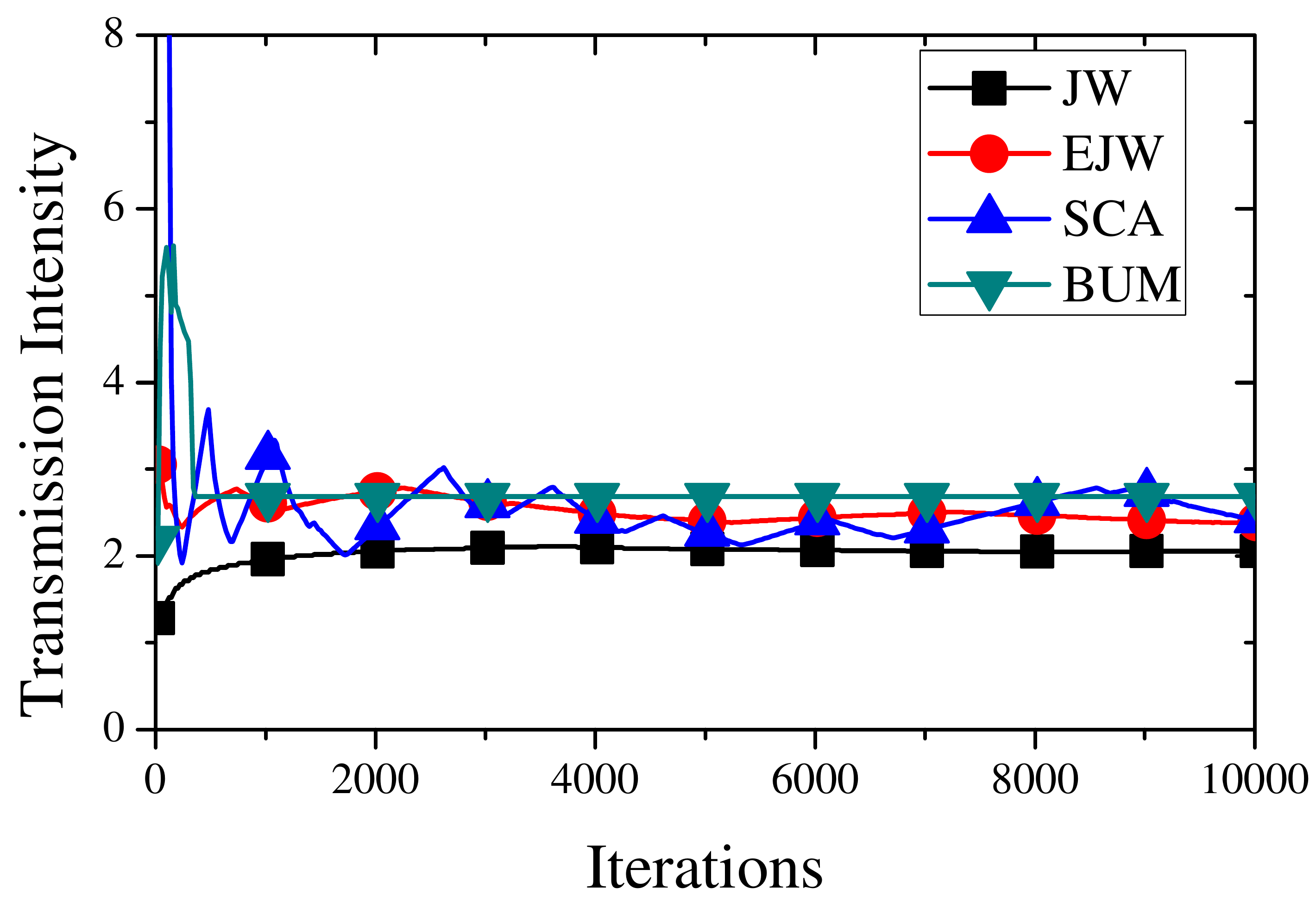}\label{fig:tr_grid} }
  \vspace{-0.2cm}
    \caption{Trace of transmission intensity}
    \label{fig:tr_inten}
  \end{center}
  \vspace{-0.2cm}
\end{figure*}

\subsection{Utility Maximization}
\noindent{\fbs Convergence speed.} Fig~\ref{fig:tr_inten} shows the
transmission intensity where the graph structure is tree. Note that in
tree graphs, all of the algorithms have to converge to the same point,
because $e_B(\by) = 0$ for all $\by$ when the graph is tree. In the
results, {\bf BUM} becomes stable within only 1000 iterations, whereas
the other algorithms does not converge until 10000
iterations. Although the lines of JW and EJW seems to be converged,
they grow up very slowly. For the other interference graphs, the trace
patterns look similar with the trace of tree graph. All of the
algorithms do not converge until 10000 iterations except {\bf BUM}
which converges within 1000 iterations for all graphs. 
In this
  simulation, we assume that each update of {\bf BUM} spends a time
  slot for one packet transmission.
Indeed, since each node $i$ broadcasts just $y_i(t)$ at each update,
{\bf BUM} does not need the entire time slot. Thus, we can
use {\bf BUM} as an offline algorithm to find the initial transmission
intensities so that the network utility becomes very close to the
maximum network utility at the beginning. 

\smallskip
\noindent{\fbs Network utility.} As we stated in
Theorem~\ref{errorofBUM}, {\bf BUM} generates error due to the Bethe
approximation on intensity update. However, the error is not
significant in our test scenarios.  By numerical analysis, we
get the network utility when {\bf BUM} is used:-19.9 (for a $5\times5$
grid interference graph) and -8.1 (for a complete interference graph
links). The utility is close to that from the conventional
algorithms based on MCMC: -20.6 (for a $5\times5$ grid interference
graph) and -8.05 (for a complete interference graph with 5 links). For
the star graph with 5 links, all of the algorithms converge to
-3.3. We found that all of the algorithms achieve similar utilities,
while {\bf BUM} converges much faster than prior algorithms.







\section{Conclusions}

Recently, throughput and utility optimal CSMA algorithms are
proposed. The simple and distributed MAC protocol can achieve the both
throughput and utility optimal with just locally controlling of
parameters. In the previous algorithms, links iteratively update their
parameters by their own empirical service and arrival rates. However,
their convergence speed is often slow because of the stochastic
behavior of scheduling. In this paper, we firstly connect Bethe Free
Energy (BFE) with CSMA so as to dramatically reduce the convergence
speed. The motivation of this work is that the estimation on the
service can be replaced by finding maximum point of the Bethe free
energy function since the maximum point gives a good estimation on the
service rate. From this motivation, we propose an algorithm by which
the CSMA parameters can be nearly optimal without the investigation on
service rate when links know the arrival rate of neighbor links by
message exchange. In view of network utility, we
propose an utility-maximizing algorithm {\bf BUM} based on the
intensity update algorithm using BFE. Since the algorithm does not use
empirical values, {\bf BUM} provably converges in polynomial time, where
such a guarantee cannot be achievable via prior known schemes.


\bibliographystyle{IEEEtran}
\bibliography{reference}

\begin{thebibliography}{10}
\providecommand{\url}[1]{#1}
\csname url@samestyle\endcsname
\providecommand{\newblock}{\relax}
\providecommand{\bibinfo}[2]{#2}
\providecommand{\BIBentrySTDinterwordspacing}{\spaceskip=0pt\relax}
\providecommand{\BIBentryALTinterwordstretchfactor}{4}
\providecommand{\BIBentryALTinterwordspacing}{\spaceskip=\fontdimen2\font plus
\BIBentryALTinterwordstretchfactor\fontdimen3\font minus
  \fontdimen4\font\relax}
\providecommand{\BIBforeignlanguage}[2]{{%
\expandafter\ifx\csname l@#1\endcsname\relax
\typeout{** WARNING: IEEEtran.bst: No hyphenation pattern has been}%
\typeout{** loaded for the language `#1'. Using the pattern for}%
\typeout{** the default language instead.}%
\else
\language=\csname l@#1\endcsname
\fi
#2}}
\providecommand{\BIBdecl}{\relax}
\BIBdecl

\bibitem{Yedidia05constructingfree}
J.~S. Yedidia, W.~T. Freeman, and Y.~Weiss, ``Constructing free energy
  approximations and generalized belief propagation algorithms,'' \emph{IEEE
  Transactions on Information Theory}, vol.~51, pp. 2282--2312, 2005.

\bibitem{JSSW10DRA}
L.~Jiang, D.~Shah, J.~Shin, and J.~Walrand, ``Distributed random access
  algorithm: Scheduling and congestion control,'' \emph{IEEE Transactions on
  Information Theory}, vol.~56, no.~12, pp. 6182 --6207, dec. 2010.

\bibitem{JW10DC}
L.~Jiang and J.~Walrand, ``A distributed {CSMA} algorithm for throughput and
  utility maximization in wireless networks,'' \emph{IEEE/ACM Transactions on
  Networking}, vol.~18, no.~3, pp. 960 --972, June 2010.

\bibitem{IEEEpaper:Liu_Yi_Proutiere_Chiang_Poor_2009}
J.~Liu, Y.~Yi, A.~Proutiere, M.~Chiang, and H.~V. Poor, ``Towards
  utility-optimal random access without message passing,'' \emph{Wiley Journal
  of Wireless Communications and Mobile Computing}, vol.~10, no.~1, pp.
  115--128, Jan. 2010.

\bibitem{HP12SO}
N.~Hegde and A.~Proutiere, ``Simulation-based optimization algorithms with
  applications to dynamic spectrum access,'' in \emph{Proceedings of CISS},
  2012.

\bibitem{taseph92}
L.~Tassiulas and A.~Ephremides, ``Stability properties of constrained queueing
  systems and scheduling for maximum throughput in multihop radio networks,''
  \emph{IEEE Transactions on Automatic Control}, vol.~37, no.~12, pp.
  1936--1949, December 1992.

\bibitem{LSS06}
X.~Lin, N.~B. Shroff, and R.~Srikant, ``A tutorial on cross-layer design in
  wireless networks,'' \emph{IEEE Journal on Selected Areas in Communications},
  vol.~24, pp. 1452--1463, 2006.

\bibitem{yi2008stochastic}
Y.~Yi and M.~Chiang, ``Stochastic network utility maximisation: a tribute to
  {Kelly's} paper published in this journal a decade ago,'' \emph{European
  Transactions on Telecommunications}, vol.~19, no.~4, pp. 421--442, 2008.

\bibitem{YY12OC}
S.-Y. Yun, Y.~Yi, J.~Shin, and D.~Y. Eun, ``Optimal {CSMA:} a survey,'' in
  \emph{Proceedings of ICCS}, 2012.

\bibitem{JLN11FM}
L.~Jiang, M.~Leconte, J.~Ni, R.~Srikant, and J.~Walrand, ``Fast mixing of
  parallel {Glauber} dynamics and low-delay {CSMA} scheduling,'' in
  \emph{Proceedings of infocom}, april 2011, pp. 371 --375.

\bibitem{kindermann1980markov}
\BIBentryALTinterwordspacing
R.~Kindermann, J.~Snell, and A.~M. Society, \emph{Markov random fields and
  their applications}, ser. Contemporary mathematics.\hskip 1em plus 0.5em
  minus 0.4em\relax American Mathematical Society, 1980. [Online]. Available:
  \url{http://books.google.com/books?id=NeVQAAAAMAAJ}
\BIBentrySTDinterwordspacing

\bibitem{F01CG}
J.~G.~David~Forney, ``Codes on graphs: News and views,'' in \emph{Conference on
  Information Sciences and Systems}, 2001.

\bibitem{MWJ99LB}
Y.~W. K.~P.~Murphy and M.~Jordan, ``Loopy belief propagation for approximate
  inference: an empirical study,'' in \emph{In Proceedings of Uncertainty in
  Artificial Intelligence}, 1999.

\bibitem{mowa00}
J.~Mo and J.~Walrand, ``Fair end-to-end window-based congestion control,''
  \emph{IEEE/ACM Transactions on Networking}, vol.~8, no.~5, pp. 556--567,
  2000.

\bibitem{KL12ABP}
C.~H. Kai and S.~C. Liew, ``Applications of belief propagation in {CSMA}
  wireless networks,'' \emph{IEEE/ACM Transactions on Networking}, vol.~20,
  no.~4, pp. 1276--1289, 2012.

\bibitem{TDMA_thesis}
P.~Djukic, ``Scheduling algorithms for {TDMA} wireless multihop networks,''
  Ph.D. dissertation, University of Toronto, 2008.

\bibitem{yivesh06a}
Y.~Yi, G.~de~Veciana, and S.~Shakkottai, ``Learning contention patterns and
  adapting to load/topology changes in in a {MAC} scheduling algorithm,'' in
  \emph{Proceedings of IEEE WiMesh}, 2006.

\bibitem{georgii1988gibbs}
\BIBentryALTinterwordspacing
H.~Georgii, \emph{Gibbs Measures and Phase Transitions}, ser. De Gruyter
  Studies in Mathematics.\hskip 1em plus 0.5em minus 0.4em\relax W. de Gruyter,
  1988, no. V. 9. [Online]. Available:
  \url{http://books.google.com/books?id=3YdI0yww12QC}
\BIBentrySTDinterwordspacing

\bibitem{chandrasekaran2011counting}
V.~Chandrasekaran, M.~Chertkov, D.~Gamarnik, D.~Shah, and J.~Shin, ``Counting
  independent sets using the bethe approximation,'' \emph{SIAM Journal on
  Discrete Mathematics}, vol.~25, no.~2, pp. 1012--1034, 2011.

\end{thebibliography}

\newpage
\appendix
\subsection{Notations} \label{apd:nota}
Table~\ref{tab:notation} contains notations used in this paper.
\begin{table}
  \centering
  \caption{Notations}
    \begin{tabular}{|c|p{12cm}|}
\multicolumn{2}{c}{Network model}\\
\hline \hline
$V$& the set of vertices (nodes)\\
\hline
$n$& the number of vertices (nodes)\\
\hline
$E$& the set of edges such that $(i,j)\in E$ if their
transmissions interfere with each other \\
\hline
$G=(V,E)$& the interference graph \\
\hline
$\mathcal{N}(i)$& $\{j:(i,j)\in E \}$, the set of the neighboring
links of link $i$\\
\hline
$\bsigma (t) $& $[\sigma_i (t)\in \{ 0,1\}:i\in V]$, the scheduling
vector at time $t$\\
\hline 
$\mathcal{I}(G)$& $\{\bsigma \in \{0,1\}^{n}:\sigma_i+\sigma_j \le 1,
\forall (i,j) \in E \}$, the set of all feasible schedule vectors \\
\hline
$C(G)$& $\left\{ \sum_{\bsigma \in
    \sI(G)}\alpha_{\bsigma}\bsigma :\sum_{\bsigma \in
    \sI(G)}\alpha_{\bsigma}=1,~ \alpha_{\bsigma}\ge 0,~ \forall \bsigma
  \in \sI(G) \right\}$, the set of all possible service rate vectors \\
\hline
$\br$& $[r_i: i\in V]$, the transmission intensity vector\\
\hline
$s_i (\br)$& the service rate of link $i$ under CSMA with transmission
intensity vector $\br$ \\ \hline
$\lambda_i$& the packet arrival rate at link $i$\\ \hline
\multicolumn{2}{c}{Free energies}\\
\hline \hline
$F_G(\cdot;\br)$, $H_G(\cdot)$& Gibbs free energy function with intensity
vector $\br$ and Gibbs entropy (they are functions of probability
distributions on space $\mathcal{I}(G)$)  \\ \hline
$F_B (\cdot;\br)$, $H_B(~\cdot~)$& Bethe
free energy function with intensity vector $\br$ and Bethe entropy  \\
\hline
$D_B$& $\{\by: y_i\geq 0, y_i+y_j\leq 1,~\mbox{for all}~(i,j)\in E\}$, the domain of $F_B$ and $H_B$ \\\hline
$e_B(\br)$& Bethe error (refer to Definition~\ref{def:bethe}) \\
\hline
\multicolumn{2}{c}{Utilities}\\
\hline \hline
$U_i(\cdot)$& the utility function of link $i$ \\ \hline
$K_B(\by)$, $\by \in D_B$& $\beta \cdot \sum_{i \in \slink} U_i
(y_i) +H_B (\by)$, the objective function of BUM \\\hline
    \end{tabular}
  \label{tab:notation}
\end{table}

\subsection{Proof of Lemma~\ref{LEM:CONCAVE}}\label{sec:pfLEM:CONCAVE}
Let $\mathcal{H}(\by)$ denote the Hessian matrix of $K_B
  (\by)$ and $\mathcal{H}(\by)_{ij}$ denote the element of
  $\mathcal{H}(\by)$ on $i$-th row and $j$-th column. When the Hessian
  matrix $\mathcal{H}(\by)$ is negative definite ($\ie~~ \bx
  \cdot \mathcal{H}(\by) \cdot \bx \leq 0$ for all $\bx$) for all
  feasible $\by,$ $K_B(\by)$ is concave. Therefore, we will show
  the concaveness of $K_B(\by)$ by showing that $\bx \cdot
  \mathcal{H}(\by) \cdot \bx \leq 0$  for all $\bx .$

The diagonal elements $\mathcal{H}(\by)_{ii}$
  are computed as follows:
\begin{align*}
\mathcal{H}(\by)_{ii} = &
\beta\cdot U_i''(y_i)+(d(i)-1) \frac{1}{1-y_i}-\frac{1}{y_i}-\sum_{j \in
\n(i)}\frac{1}{1-y_i-y_j} \cr
=&-\alpha \beta \cdot y_i^{-\alpha-1}- \frac{1}{y_i}
-\frac{1}{1-y_i} -\sum_{j \in
\n(i)} \left( \frac{1}{1-y_i-y_j} - \frac{1}{1-y_i}\right),
\end{align*}
which is bounded above as follows :
\begin{align*}
\mathcal{H}(\by)_{ii}&  <~-\sum_{j \in
\n(i)} \left( \frac{1}{1-y_i-y_j} - \frac{1}{1-y_i} \right) \cr
&~=~-\sum_{j \in \n(i)} \left(\frac{y_j}{1-y_i}\cdot \frac{1}{1-y_i-y_j}
\right),
\end{align*}
since $-\alpha \beta \cdot y_i^{-\alpha-1}- \frac{1}{y_i}
-\frac{1}{1-y_i}<0.$ Moreover, when $y_i < 1/2,$
we can get more tight bound as follows:
\begin{align*}
\mathcal{H}(\by)_{ii} <& - 2d \cdot y_i^{-\alpha-1}+(d(i)-1) \frac{1}{1-y_i}-\frac{1}{y_i}-\sum_{j \in
\n(i)}\frac{1}{1-y_i-y_j} \cr
\stackrel{(a)}{<} &- d \cdot y_i^{-\alpha-1}-\sum_{j \in
\n(i)}  \frac{1}{1-y_i-y_j} \cr
\stackrel{(b)}{<}&-\sum_{j \in \n(i)} \left(\frac{1}{y_i} + \frac{1}{1-y_i-y_j}
\right) \cr
=&-\sum_{j \in \n(i)}\left(\frac{1-y_j}{y_i}\cdot \frac{1}{1-y_i-y_j}
\right),
\end{align*}
where for $(a)$ we use that $y_i^{-\alpha-1} > \frac{1}{1-y_i}$ when $y_i <
1/2$ and $(b)$ follows from $y_i^{-\alpha-1} > 1/y_i.$

One can easily compute the non-diagonal elements such that
\begin{equation*}
\mathcal{H}(\by)_{ij} = \mathcal{H}(\by)_{ji}=\begin{cases}
-\frac{1}{1-y_i-y_j}<0 & \mbox{if}~~(i,j)\in E \cr
0 & \mbox{otherwise}.\end{cases}
\end{equation*}

Without loss of generality, let $y_u \le y_v$ when the edge is denoted by
$(u,v).$ Then,
\begin{align*}
&\bx^T  \mathcal{H}(\by)\bx = \sum_{i \in \slink} x_i^2
\mathcal{H}(\by)_{ii}+ \sum_{(i,j) \in E} 2x_i x_j
\mathcal{H}(\by)_{ij}\cr
&< -\sum_{i \in \slink:y_i < \frac{1}{2}} \sum_{j \in \N(i)}
\frac{1-y_j}{y_i} \frac{x_i^2}{1-y_i-y_j} - \sum_{i \in \slink : y_i \geq
  \frac{1}{2}} \sum_{j \in \N(i)} \frac{y_j}{1-y_i}
\frac{x_i^2}{1-y_i-y_j}   -\sum_{(i,j) \in E}\frac{2x_i x_j}{1-y_i-y_j} \cr
&< -\sum_{i \in \slink} \sum_{j \in \N(i): y_i \le y_j}
\frac{1-y_j}{y_i} \frac{x_i^2}{1-y_i-y_j} - \sum_{j \in \slink }
\sum_{i \in \N(j) : y_i < y_j} \frac{y_i}{1-y_j}
\frac{x_j^2}{1-y_i-y_j}  -\sum_{(i,j) \in E}\frac{2x_i x_j}{1-y_i-y_j}\cr
&= -\sum_{(i,j) \in E} \Big( \frac{1-y_j}{y_i}\cdot \frac{1}{1-y_i-y_j}
 x_i^2 + \frac{2x_i x_j}{1-y_i-y_j} +  \frac{y_i}{1-y_j}\cdot
\frac{1}{1-y_i-y_j} x_j^2 \Big) \cr
&=  -\sum_{(i,j) \in E} \frac{1}{1-y_i-y_j}\Big( \sqrt{\frac{1-y_j}{y_i}}
 x_i + \sqrt{\frac{y_i}{1-y_j}}x_j \Big)^2 ~ \le ~ 0.
\end{align*}
Therefore, $\mathcal{H}$ is negative definite matrix.



\subsection{Proof of Lemma~\ref{lem:betheupdate}}
Recall that
\begin{align*}
  \delta_2&:=\min\left \{ c_2(t_*), \frac{1}{2
      (\exp(\beta2^{\alpha})+1)}\right \}, \quad\quad \delta_1:=\min
  \left\{c_1(t_*)
    ,\frac{\beta2^{\alpha}\delta_2^d}{4(1+\beta2^{\alpha}d\delta_2^{d-1})}
  \right\}, \cr \delta_3&:=\min\left \{ c_2(t_*),
    \frac{\delta_1}{2 \exp(\beta \delta_1^{-\alpha})}
  \right\},~\mbox{and}~ \cr
t_* & := \inf \left\{ \tau :  \frac{1}{\sqrt{t}}\left|\frac{\partial K_B (\by(t))}{\partial
    y_i}\right| < \frac{1}{2}\min\{c_1(t),c_2(t)\} ~~\forall~ t \geq\tau\right\}.\end{align*} 
In this proof, for notational simplicity, we
introduce 
$\varepsilon_1 :=
\frac{\beta2^{\alpha}\delta_2^d}{4(1+\beta2^{\alpha}d\delta_2^{d-1})}$,
$\varepsilon_2:=\frac{1}{2 (\exp(\beta 2^{\alpha})+1)}$, and
$\varepsilon_3 :=\frac{\delta_1}{2\exp(\beta\delta_1^{-\alpha})}$.

We start by stating three key lemmas which play key roles in the proof
of Lemma \ref{lem:betheupdate}. First, by Lemma~\ref{LEM:TSTAR}, the gradient
of $K_B(\by(t))$ is bounded above with $\frac{1}{2}\min\{c_1(t),c_2(t)\}$ after
time $t^*.$ Next, we show that $\by(t+1)$ goes away from the boundary
of $D^*_B$ when $\by(t)$ is within $2\min\{c_1(t),c_2(t)\}$ away from the
boundary, by Lemma~\ref{lem:delta2}, Lemma~\ref{lem:delta3}, and
Lemma~\ref{lem:b3}. Then, the update of $\by(t)$ does not hit the
boundary of $D^*_B$ always.
\begin{lemma}\label{LEM:TSTAR}
There exists $t_*$ such that
, for every link $i$
$$\frac{1}{\sqrt{t}}\left|\frac{\partial K_B (\by(t))}{\partial
    y_i}\right| < \frac{1}{2}\min\{c_1(t),c_2(t)\} ,~~\forall~t \ge t_*.$$
\end{lemma}
\begin{proof}
To conclude this proof, we will show that
$$ \lim_{t \rightarrow \infty} \frac{1}{\min\{c_1(t),c_2(t)\}
  \sqrt{t}}\frac{\partial K_B(\by(t))}{\partial y_i} = 0.$$
The proof starts from the range of first derivative function at time $t$:
\begin{align*}
\frac{\partial K_B (\by(t))}{\partial y_i} =& U'_{i}(y_i(t))-(d(i)-1)\log(1-y_i(t))-\log y_i(t) +\sum_{j\in \n(i)}\log(1-y_i(t)-y_j(t))\cr
\leq &\beta  y_i(t)^{-\alpha} -\log \frac{y_i(t)}{1-y_i(t)}+\sum_{j\in \n(i)}\log\frac{1-y_i(t)-y_j(t)}{1-y_i(t)} \cr
\leq &\beta  y_i(t)^{-\alpha} -\log \frac{y_i(t)}{1-y_i(t)} \cr
\leq &\big( c_1 (t) \big)^{-\alpha}-\log(c_1(t)),
\end{align*}
where the last inequality stems from the fact that $y_i(t) \ge c_1(t).$
Therefore, from \eqref{eq:c1c2},
$$\lim_{t \rightarrow \infty} \frac{1}{
  \sqrt{t}}\frac{1}{\min\{c_1(t),c_2(t)\}}\frac{\partial K_B(\by(t))}{\partial y_i} \le 0. $$
Now, the remaining part is $\lim_{t \rightarrow \infty} \frac{1}{\min\{c_1(t),c_2(t)\}
  \sqrt{t}}\frac{\partial K_B(\by(t))}{\partial y_i} \ge 0. $ Since
$1-y_i(t)-y_j(t) \ge c_2(t)$ for all $i\neq j$,
\begin{align*}
\frac{\partial K_B (\by(t))}{\partial y_i}&=U'_{i}(y_i(t))-(d(i)-1)\log(1-y_i(t))-\log y_i(t) +\sum_{j\in \n(i)}\log(1-y_i(t)-y_j(t))\cr
&\geq \sum_{j\in \n(i)}\log(1-y_i(t)-y_j(t)) \geq d(i)\log(c_2(t)).
\end{align*}
Therefore, from \eqref{eq:c1c2},
$$\lim_{t \rightarrow \infty} \frac{1}{\min\{c_1(t) ,c_2(t)\} \sqrt{t}}\frac{\partial K_B(\by(t))}{\partial y_i}\ge 0. $$
\end{proof}

\begin{lemma}\label{lem:delta2}If $y_i\geq 1-
2\varepsilon_2$ and $\y\in D^*_B$,
$\frac{\partial K_B(\by)}{\partial y_i}< 0$. 
\end{lemma}
\begin{proof}
\begin{align*}
\frac{\partial K_B(\by)}{\partial y_i}&
=\beta  y_i^{-\alpha} -(d(i)-1)\log(1-y_i)-\log y_i +\sum_{j\in \n(i)}\log(1-y_i-y_j)\cr
&=\beta  y_i^{-\alpha}-\log \frac{y_i}{1-y_i}+\sum_{j\in \n(i)}\log\frac{1-y_i-y_j}{1-y_i}\cr
&<\beta  y_i^{-\alpha}-\log \frac{y_i}{1-y_i}\cr
 &\leq \beta  (\frac{1}{2})^{-\alpha}-\log
\frac{1-2\varepsilon_2}{2\varepsilon_2}\leq 0,
\end{align*}
where the last inequality is from our choice of
$$\varepsilon_2=\frac{1}{2 (\exp(\beta  2^{\alpha})+1)}.$$
\end{proof}

\begin{lemma}\label{lem:delta3}If $y_i\leq
2\varepsilon_1$ and $\y\in D^*_B$,
$\frac{\partial K_B(\by)}{\partial y_i}> 0$. 
\end{lemma}
\begin{proof}
\begin{align*}
\frac{\partial K_B(\by)}{\partial y_i}&
=\beta  y_i^{-\alpha}-(d(i)-1)\log(1-y_i)-\log y_i +\sum_{j\in \n(i)}\log(1-y_i-y_j)\cr
&> \beta  y_i^{-\alpha}-\log y_i +d\,\log\left(\delta_2-y_i\right)\cr
&=\log \frac{\exp(\beta  y_i^{-\alpha}) \left(\delta_2-y_i\right)^d}{y_i}\cr
&=\log \frac{\exp(\beta  y_i^{-\alpha}) \delta_2^d\left(1-\frac{y_i}{\delta_2}\right)^d}{y_i}\geq 0,
\end{align*}
where the last inequality stems from the fact that $y_i\leq 2\varepsilon_1$ with
our choice of $\varepsilon_1 =\frac{\beta2^{\alpha}\delta_2^d}{4(1+\beta2^{\alpha}d\delta_2^{d-1})}$
as
\begin{align*}
\frac{\beta y_i^{-\alpha}
\delta_2^d\left(1-\frac{y_i}{\delta_2}\right)^d}{y_i} \geq
\frac{\beta y_i^{-\alpha}
\delta_2^d\left(1-d\frac{y_i}{\delta_2}\right)}{y_i}
\geq
\frac{\beta 2^{\alpha}
\delta_2^d\left(1-d\frac{y_i}{\delta_2}\right)}{y_i}\geq 1
\end{align*}
and $\exp(\beta  y_i^{-\alpha}) \ge \beta y_i^{-\alpha}$ since $\beta
y_i^{-\alpha} \ge 1$.
\end{proof}

\begin{lemma}\label{lem:b3}If $y_i+y_j\geq
1-2\varepsilon_3$ and $\y\in D^*_B$,
$\frac{\partial K_B(\by)}{\partial y_i}< 0$. 
\end{lemma}
\begin{proof}
\begin{align*}
\frac{\partial K_B(\by)}{\partial y_i}&=
\beta  y_i^{-\alpha}- (d(i)-1)\log(1-y_i)-\log y_i +\sum_{k\in \n(i)}\log(1-y_k-y_i)\cr
&=\beta  y_i^{-\alpha}-\log {y_i}+\log(1-y_i-y_j)+\sum_{k\in \n(i)\setminus j}\log\frac{1-y_k-y_i}{1-y_i}\cr
&<\beta  y_i^{-\alpha}-\log {y_i}+\log(1-y_i-y_j)\cr
&\leq \beta  y_i^{-\alpha}-\log \delta_1+\log 2\varepsilon_3 \leq 0,
\end{align*} where the last inequality is from our choice of
$\varepsilon_3=\frac{\delta_1}{2\exp(\beta \delta_1^{-\alpha})}.$
\end{proof}
\smallskip
\noindent{\fbs Completing the proof of Lemma~\ref{lem:betheupdate}.} 
For proving $\by (t) \in D^*_B$, we need the following three inequalities: 
\begin{eqnarray}
y_i (t)& <& 1-\delta_2\label{eq1}\\ 
y_i (t) &>&\delta_1\label{eq2}\\
y_i (t) + y_j (t)& <&1-\delta_3.\label{eq3}
\end{eqnarray}

\vspace{0.2cm}

\noindent{\bf \em Proof of \eqref{eq1}.} 
Let $t_2 := c_2^{-1}(\delta_2).$ Then, for time $t<t_2,$ $y_i(t)
< 1-\delta_2$ from the dynamic bound. For time $t\ge t_2,$ $y_i(t) <
1-\delta_2,$ since $\frac{1}{\sqrt{t}}\left|\frac{\partial
    K_B(\by)}{\partial y_i}\right| < \frac{c_2(t)}{2}\le \frac{\delta_2}{2}$ from Lemma \ref{LEM:TSTAR} and
$\frac{\partial K_B(\by)}{\partial y_i}< 0$ if $y_i > 1-2\delta_2,$ from
Lemma \ref{lem:delta2}.

\vspace{0.2cm}

\noindent{\bf \em Proof of \eqref{eq2}.} 
Similarly, let $t_1 := c_1^{-1}(\delta_1).$ Then, for time $t<t_1,$ $y_i(t)
> \delta_1$ from the dynamic bound. For time $t\ge t_1,$ $y_i(t) >
\delta_1,$ since $\frac{1}{\sqrt{t}}\left|\frac{\partial
    K_B(\by)}{\partial y_i}\right| < \frac{c_1(t)}{2} < \frac{\delta_1}{2}  $ from Lemma \ref{LEM:TSTAR} and
$\frac{\partial K_B(\by)}{\partial y_i} > 0$ if $y_i < 2\delta_1,$ from
Lemma \ref{lem:delta3}.

\vspace{0.2cm}

\noindent{\bf \em Proof of \eqref{eq3}.} 
Let $t_3 := c_2^{-1}(\delta_3).$ Then, for time $t<t_3,$ $y_i(t)
+y_j(t)< 1-\delta_3$ from the dynamic bound. For time $t\ge t_3,$
$y_i(t) +y_j (t)<
1-\delta_3,$ since $\frac{1}{\sqrt{t}}\left(\left|\frac{\partial
    K_B(\by)}{\partial y_i}\right| + \left|\frac{\partial
    K_B(\by)}{\partial y_j}\right| \right)<c_2(t) \le \delta_3  $ from Lemma \ref{LEM:TSTAR} and
$\max\{\frac{\partial K_B(\by)}{\partial y_i},\frac{\partial K_B(\by)}{\partial
  y_i}\} < 0$ if $y_i + y_j > 1-2\delta_3,$ from
Lemma \ref{lem:b3}.

\vspace{0.2cm}

By combining \eqref{eq1}, \eqref{eq2} and \eqref{eq3}, it follows that
$\by (t) \in D^*_B$ for all $t.$

\end{document}